%% file: ms.tex
\documentclass[12pt]{l4dc2023}

\title[Learning Coherent Clusters in Weakly-Connected Network Systems]{Learning Coherent Clusters in Weakly-Connected Network Systems}
\usepackage{times}
\usepackage{mysty}
\usepackage{bbold,amsmath,amsfonts}
\allowdisplaybreaks
\usepackage{cite}
\usepackage{ifthen}
\usepackage{algorithm}

\input{edits}

\newboolean{archive}
\setboolean{archive}{true}

\usepackage{tikz}
\usetikzlibrary{shapes, arrows}
\usepackage{enumitem}
\tikzstyle{block} = [draw, rectangle, minimum height=3em, minimum width=3em]
\tikzstyle{sum} = [draw, circle, node distance=1cm]
\tikzstyle{input} = [coordinate] \tikzstyle{output} = [coordinate]
\tikzstyle{tmp} = [coordinate]

\newcommand{\tr}{\text{tr}}

\newif\ifshownotes
\shownotestrue
\definecolor{notetext}{rgb}{0.7,0,0}

\newtheorem{assumption}{Assumption}

\author{%
 \Name{Hancheng Min} \Email{hanchmin@jhu.edu}\\
 \Name{Enrique Mallada} \Email{mallada@jhu.edu}\\
 \addr 
 Department of Electrical and Computer Engineering\\
 Johns Hopkins University, Baltimore, MD, U.S.
}

\begin{document}

\maketitle

\begin{abstract}%
    We propose a structure-preserving model-reduction methodology for large-scale dynamic networks with tightly-connected components. First, the coherent groups are identified by a spectral clustering algorithm on the graph Laplacian matrix that models the network feedback. Then, a reduced network is built, where each node represents the aggregate dynamics of each coherent group, and the reduced network captures the dynamic coupling between the groups. We provide an upper bound on the approximation error when the network graph is randomly generated from a weight stochastic block model. Finally, numerical experiments align with and validate our theoretical findings.
\end{abstract}


\section{Introduction}
In networked dynamical systems, coherence refers to a coordinated behavior from a group of nodes such that all nodes have similar dynamical responses to some external disturbances~\citep{chow1982time}. Coherence analysis is useful in understanding the collective behavior of large networks, including consensus networks~\citep{Olfati-Saber20041520}, transportation networks~\citep{Bamieh2012}, and power networks~\citep{ramaswamy1995}. However, little do we know about the underlying mechanism that causes such coherent behavior to emerge in various networks.

Classic slow coherence analyses~\citep{chow1982time,ramaswamy1996,romeres2013,tyuryukanov2021,fritzsch2022} (with applications mostly to power networks) usually consider the second-order electro-mechanical model without damping: $\ddot{x}=-M^{-1}Lx$, where $M$ is the diagonal matrix of machine inertias, and $L$ is the Laplacian matrix whose elements are synchronizing coefficients between pair of machines. The coherency or synchrony~\citep{ramaswamy1996} (a generalized notion of coherency) is identified by studying the first few slowest eigenmodes (eigenvectors with small eigenvalues) of $M^{-1}L$. The analysis can be carried over to the case of uniform~\citep{chow1982time} and non-uniform~\citep{romeres2013} damping. However, such state-space-based analysis is limited to very specific node dynamics (second order) and does not account for more complex dynamics or controllers that are usually present at a node level; e.g., in the power systems literature~\citep{jpm2021tac, jbvm2021lcss, ekomwenrenren2021}.  
There is, therefore, the need for coherence identification procedures that work for more general network systems.
    
    Recently, it has been theoretically established that coherence naturally emerges when the connectivity of a group of nodes is sufficiently large, regardless of the node dynamics, as long as the interconnection remains stable~\citep{min2019cdc,min2021a}. The analysis also provides an asymptotically (as the network connectivity increases) exact characterization of the coherent response, which amounts to a harmonic sum of individual node transfer functions. Thus, in a sense, coherence identification is closely related to the problem of finding tightly connected components in the network, for which many clustering algorithms based on the spectral embedding of graph adjacency or Laplacian matrices, exist and are theoretically justified~\citep{bach2004learning}. This leads to the natural question: \emph{Can these graph-based clustering algorithms be adopted for coherence identification in networked dynamical systems?} Intuitively, when we apply those clustering algorithms to identify tightly-connected components in the network, each component should be coherent also in the dynamical sense. Then, applying ~\citet{min2019cdc,min2021a} for each cluster should lead to a good model for each coherent group, which, after interconnected with an appropriately chosen reduced graph, should lead to a good network-reduced aggregate model of the dynamic interactions across coherent components.~\citet{min2022} formalizes such an approach exclusively for networks with two coherent components.

    In this paper, we extend the result in~\citet{min2022} to networks with an arbitrary number of coherent groups. Specifically, our structure-preserving approximation model for large-scale networks is constructed in two stages: First, the coherent groups are identified by a spectral clustering algorithm solely on the graph Laplacian matrix of the network; Then a reduced network, in which each node represents the aggregate dynamics of one coherent group, approximates the dynamical interactions between the coherent groups in the original network. More importantly, we provide an upper bound on the approximation error when the network graph is randomly generated from a weight stochastic block model, and the numerical results align with our theoretical findings.

    Structure-preserving model reduction has been mostly studied for mechanical systems~\citep{Li2006,LALL2003304} using Krylov subspace projection, and has only recently been adopted for network systems such as power networks~\citep{Bita2021}. However,~\citet{Bita2021} assumes second-order nodal dynamics, and the resulting model can not be interpreted as a network. Our approach exploits the natural multi-cluster structure of many network systems, resulting in a reduced network that captures the interaction among the clusters.

    The rest of the paper is organized as follows: We formalize the coherence identification problem in Section \ref{sec_prelim} and also propose our reduction algorithm. Then we show in Section \ref{sec_ideal_net_model} the rationale behind the algorithm and provide theoretical justification in Section \ref{sec_analysis}. Lastly, we validate our model by numerical experiments in Section \ref{sec_num}. 
    
    \emph{Notation:}~For a real vector $x$, $\|x\|=\sqrt{x^Tx}$ denotes the $2$-norm of $x$, $[x]_i$ denotes its $i$-th entry, and for a real matrix $A$, $\|A\|,\|A\|_F$ denotes the spectral norm, and the Frobenius norm, respectively. For a Hermitian matrix $A$ of size $n$, we let $\lambda_i(A)$ denote its $i$-th smallest eigenvalue, and $v_i(A)$ the associated unit-norm eigenvector. We let $\dg\{x_i\}_{i=1}^n$ denote a $n\times n$ diagonal matrix with diagonal entries $x_i$, $I_n$ denote the identity matrix of order $n$, $V^T$ denote the transpose of matrix $V$, $\one_n$ denote $[1,\cdots,1]^T $ with dimension $n$, and  $[n]$ denote the set $\{1,2,\cdots,n\}$.  For non-negative random variables $X(n),Y(n)$, we write $X(n)\sim\mathcal{O}_p(Y(n))$ if $\exists M>0$, s.t. $\lim_{n\ra \infty}\prob\lp X(n)\leq M Y(n)\rp=1$. We write $X(n)\sim\Omega_p(Y(n))$ if $\exists M>0$, s.t. $\lim_{n\ra \infty}\prob\lp X(n)\geq M Y(n)\rp=1$.
\section{Preliminaries}\label{sec_prelim}
\subsection{Network Model}
We consider a similar network model to the one considered in~\citet{min2021a,min2022}. The network consisting of $n$ nodes ($n\geq 2$), indexed by $i\in[n]$ with the block diagram structure in Fig.\ref{fig_blk_digm}. $L$ is the Laplacian matrix of an undirected, weighted graph that describes the network interconnection. We further use $f(s)$ to denote the transfer function representing the dynamics of the network coupling, and $G(s)=\mathrm{diag}\{g_i(s)\}_{i=1}^n$ to denote the nodal dynamics, with $g_i(s),\ i\in[n]$, being an SISO transfer function representing the dynamics of node $i$. 
\begin{figure}[ht]
	\centering
	\begin{tikzpicture}[auto, node distance=1.5cm,>=latex']
	\node [input, name=input] {};
	\node [sum, right of=input] (sum) {};
	\node [block, right of=sum] (plant) {$G(s)$};
	\node [output, right of=plant, node distance=2cm] (output) {};
	\node [block, below of=plant] (laplacian) {$f(s)L$};
	\draw [draw,->] (input) -- node {$u$} (sum);
	\draw [->] (sum) -- (plant);
	\draw [->] (plant) -- node [name=y]{$y$} (output);
	\draw [->] (y) |- (laplacian);
	\draw [->] (laplacian) -| node[pos=0.95]{$-$}(sum);
		
	\end{tikzpicture}
	\caption{Block Diagram of General Networked Dynamical Systems}\label{fig_blk_digm}
\end{figure}
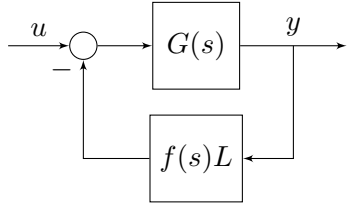

The network takes a vector signal $u=[u_1,\cdots,u_n]^T$ as input, whose component $u_i$ is the disturbance or input to node $i$. The network outputs a vector $y=[y_1,\cdots,y_n]^T$ that contains the individual node outputs $y_i,i=1,\cdots,n$. We are interested in characterizing and approximating the response of the transfer matrix $T_{yu}(s)$ under certain assumptions on the network topology, i.e., the Laplacian matrix $L$. 


\subsection{Network Coherence and Structure-preserving Model Reduction}
Recent work~\citep{min2019cdc,min2021a} has shown that, under mild assumptions, the following holds\footnote{In~\citet{min2021a}, the transfer matrix $\frac{1}{n}\bar{g}(s)\one\one^T$ appeared in the limit, where $\bar{g}(s)=\lp\frac{1}{n}\sum_{i=1}^ng_i^{-1}(s)\rp^{-1}$. It is easy to verify that $\frac{1}{n}\bar{g}(s)\one\one^T=\hat{g}(s_0)\one\one^T$ } for almost any $s_0\in\mathbb{C}$,
\be
    \lim_{\lambda_2(L)\ra\infty}\|T_{yu}(s_0)-\hat{g}(s_0)\one\one^T\|=0\,,\label{eq_coherent}
\ee
where
\be
    \hat{g}(s)=\lp \sum_{i=1}^ng_i^{-1}(s)\rp^{-1}\,.\label{eq_hat_g_def}
\ee
That is, when the algebraic connectivity $\lambda_2(L)$ of the network is high, one can approximate $T_{yu}(s)$ by a rank-one transfer matrix. Such a rank-one transfer matrix $\hat{g}(s_0)\one\one^T$ precisely describes the coherent behavior of the network: The network takes the aggregated input $\hat{u}=\one^Tu=\sum_{i=1}^nu_i$, and responds coherently as $\hat{y}\one$, where $\hat{y}=\hat{g}(s)\hat{u}$. Therefore, it suffices to study $\hat{g}(s)$ to understand the coherent behavior in a tightly-connected network. The aggregate dynamics $\hat{g}(s)$ has been studied for tightly-connected power networks~\citep{min2020lcss}, and~\citet{jbvm2021lcss,haberle2022grid} proposed a control design that leads to desirable response for the entire network by shaping the response of $\hat{g}(s)$.

However, practical networks are not necessarily tightly-connected. Instead, they often contain multiple groups of nodes such that within each group, the nodes are tightly-connected while between groups, the nodes are weakly-connected. Then the network dynamics can be reduced to dynamic interactions among these groups. To approximate such interaction, it is natural first to identify \emph{coherent groups}, or \emph{coherent clusters}, in the network, then apply the aforementioned analysis to obtain the coherent dynamics $\hat{g}(s)$ for each group, and replace the entire coherent group by an aggregate node with $\hat{g}(s)$. Lastly, one needs to find a reduced network of the same size as the number of coherent groups, which characterize the interaction among these groups. The aggregate dynamics and the reduced network allow us to build a network model with exactly the same structure as the one in Figure \ref{fig_blk_digm} but with a much smaller size, for which we refer to such an approach as \emph{structure-preserving model reduction} and call the resulting reduction model \emph{structure-preserving}. Figure \ref{fig_algo_approx_model} shows our proposed reduced model in the case of three coherent groups, for which the algorithm details are explained later.

In the case of two coherent groups,~\citet{min2020lcss} proposed an algorithm that first uses a simple spectral clustering algorithm
to identify the two coherent groups, then shows that the weight of the reduced network (a two-node graph with a single edge) is determined by the algebraic connectivity of the original Laplacian matrix $\lambda_2(L)$ and the size of each group. However, such an approach does not work for networks with more than two coherent groups.

\subsection{Our Algorithm}
In this paper, we propose a structure-preserving model reduction algorithm for networks with an arbitrary number of groups.
\begin{algorithm}[!h]
    \caption{Structure-Preserving Network Reduction via Spectral Clustering}
    \KwData{Network Model $\lp G(s)=\dg\{g_i(s)\}_{i=1}^n, L, f(s)\rp$; Number of clusters $k$}
    \textbf{Do}:
    \begin{enumerate}
        \item $(\{\mathcal{I}_i\}_{i=1}^k, V_k, \Lambda_k)\la \text{SpectralClustering}(L)$; \tcp{Spectral clustering}
        Construct $P_{\{\mathcal{I}_i\}_{i=1}^k}$ as in \eqref{eq_ind_mat};
        \item $\hat{g}_i(s)\la\lp\sum_{j\in\mathcal{I}_i}g_j^{-1}(s)\rp^{-1},\ i=1,\cdots,k$; \tcp{Aggregation}
        
        $\hat{G}(s)=\dg\{\hat{g}_i(s)\}_{i=1}^k$; 
        \item 
        $S\la$(Solution to \eqref{eq_V_fit_opt_fix_part}); 
        
        $L_k=(S^{-1})^T\Lambda_k S^{-1}$; \tcp{Construct reduced network}
    \end{enumerate}
    \KwResult{$\hat{T}_k(s)\la P_{\{\mathcal{I}_i\}_{i=1}^k}(I_k+\hat{G}_k(s)L_kf(s))^{-1}\hat{G}_k(s)P^T_{\{\mathcal{I}_i\}_{i=1}^k}$}
    \label{algo_approx_model}
\end{algorithm}

This algorithm, whose rationale will be explained in detail in Section \ref{sec_ideal_net_model}, follows the same procedure as we discussed in the previous section: Firstly, we utilize some spectral clustering algorithm to obtain a $k$-way partition $\{\mathcal{I}_i\}_{i=1}^n$ of $[n]$ that encodes the clustering results. Notice that here any spectral clustering algorithm works. For subsequent steps, we also need to keep the first $k$ smallest eigenvalues of $L$ (in $\Lambda_k=\dg\{\lambda_i(L)\}_{i=1}^k$) and their associated eigenvectors (in $V_k=\bmt v_1(L)&v_2(L)&\cdots v_k(L)\emt$). Then the nodes in the same group $\mathcal{I}_i$ are aggregated into $\hat{g}_i(s)$. Lastly, the Laplacian matrix of the reduced network is constructed after solving an optimization problem \eqref{eq_V_fit_opt_fix_part} that can be viewed as a refinement process on the Laplacian spectral embedding $V_k$. This algorithm will return a transfer matrix $\hat{T}_k(s)$ as an approximation model of the original transfer matrix $T_{yu}(s)$. The algorithm is illustrated in Figure \ref{fig_algo_approx_model}.

In the rest of the paper, we first discuss how our algorithm is constructed based on the aforementioned coherence analysis~\citep{min2021a} in Section \ref{sec_ideal_net_model}, then show that our proposed approximation model is asymptotically accurate in a random graph setting where the network graph is sampled from a \emph{weighted stochastic block model}~\citep{ahn2018hypergraph} by showing an approximation error bound between the network $T_{yu}(s)$ and the proposed reduced model $\hat{T}_k(s)$ (in Section \ref{sec_analysis}). Lastly, we verify our theoretical findings through a numerical simulation in Section \ref{sec_num}.
\begin{figure}[t]
\centering
\includegraphics[width=0.9\textwidth]{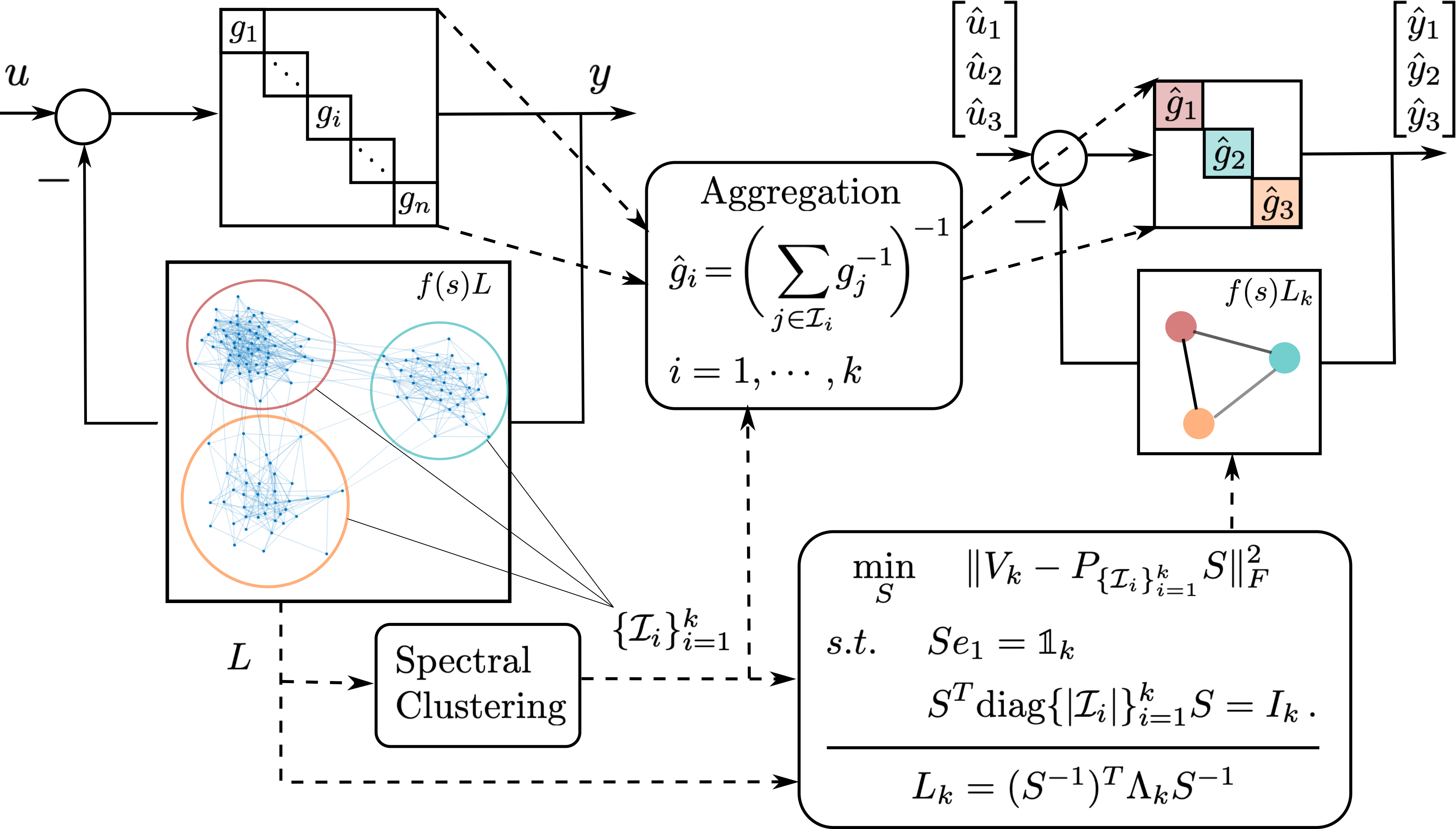}
\caption{Functional illustration of Algorithm \ref{algo_approx_model}.}
\label{fig_algo_approx_model}
\end{figure}
\section{Structure-Preserving Network Reduction via Spectral Clustering}\label{sec_ideal_net_model}
Our algorithm roots in the recent analysis~\citep{min2021a,min2022} showing that the network transfer matrix $T_{yu}(s)$ is approximately low rank for networks with Laplacian matrices satisfying some spectral property. Such a low-rank approximation is generally not structure-preserving, for which we use its closest structure-preserving approximation, obtained by spectral clustering on graph Laplacian $L$ and a refinement process on its eigenvectors $V_k$, as our final reduction model for the original $T_{yu}(s)$.
\subsection{Low-rank Approximation of Network Transfer Matrix}\label{ssec_low_rank_approx}
Given the network Laplacian $L$ and its first $k$ smallest eigenvalues (in a diagonal matrix) $\Lambda_k=\dg\{\lambda_i(L)\}_{i=1}^k$ and the associated eigenvectors $V_k=\bmt v_1(L)&v_2(L)&\cdots v_k(L)\emt$ (we also refer it as \emph{Laplacian spectral embedding}), we define the following rank-$k$ transfer matrix
\be
    T_k(s)=V_k(V_k^TG^{-1}(s)V_k+f(s)\Lambda_k)^{-1}V_k^T\,,
\ee
and we have the following result:
\begin{theorem}\label{thm_Tk}
    For $s_0\in\compl$ that is not a pole of $f(s)$ and has these two quantities $$\|T_k(s_0)\|:= M_1, \text{and }\max_{1\leq i\leq n}|g_i^{-1}(s_0)|:= M_2\,,$$ finite. Then whenever $|f(s_0)|\lambda_{k+1}(L)> M_2+M_1M_2^2$, the following inequality holds:
    \be
        \lV T_{yu}(s_0)-T_k(s_0)\rV\leq \frac{\lp M_1M_2+1\rp^2}{|f(s_0)|\lambda_{k+1}(L)-M_2-M_1M_2^2}\,.\label{eq_T_norm_bd}
    \ee
\end{theorem}
\ifthenelse{\boolean{archive}}{The proof is shown in Appendix~\ref{app_pf_thm1_3}.}{We refer the readers to~\citet{min2022b} for the proof.} Previous work presented similar approximation results for the case $k=1$~\citep{min2021a} and $k=2$~\citep{min2022}. Theorem \ref{thm_Tk} shows that in the large $\lambda_{k+1}(L)$ regime, one can somewhat approximate the original transfer matrix $T_{yu}(s)$ by a low-rank one $T_k(s)$, but the approximation result in \eqref{eq_T_norm_bd} is weaker than that the two transfer matrices $T_{yu}(s)$ and $T_k(s)$ are close in the $\mathcal{H}_\infty$ sense. It heavily depends on the choice of $s_0$, the frequency of interest, as we should not expect $T_{yu}(s)$ and $T_k(s)$ to behave similarly under input of any frequency. For the case of $k=1$,~\citet{min2021a} have shown that if $\sup_{s\in(-j\eta,+j\eta)}\|T_{yu}(s)-T_k(s)\|$ is small for some $\eta>0$, then one can show, provided that $T_{yu}(s)$ and $T_k(s)$ are stable, the time domain responses of the two transfer matrices under low-frequency inputs (characterized by $\eta$) are close to each other. 

Following such observation, we consider any $\hat{T}_k(s)$ with $\sup_{s\in(-j\eta,+j\eta)}\|T_{yu}(s)-\hat{T}_k(s)\|$ being small for some $\eta>0$ as a good approximation for the original network. Applying \eqref{eq_T_norm_bd} uniformly over $\{s: s\in(-j\eta,+j\eta)\}$, one can show that $T_k(s)$ is such a good approximation when $\lambda_{k+1}(L)$ is large. However, $T_k(s)$ is, in general, not structure-preserving, and thus may not be interpreted as a reduced network of aggregate nodes. Therefore, we need to find a structure-preserving $\hat{T}_k(s)$ that is close to $T_k(s)$.

\subsection{Structured Low-rank Approximation via Spectral Embedding Refinement}
We first discuss the case when $T_k(s)$ is structure-preserving. We show that a special property on the Laplacian spectral embedding $V_k$ suffices. For some $\mathcal{I}\subseteq [n]$, we let $\one_\mathcal{I}$ be an $n\times 1$ vector such that $[\one_\mathcal{I}]_i=\begin{cases}
    1,& i\in \mathcal{I}\\
    0, & i\notin \mathcal{I}
\end{cases}$.
\begin{definition}
        A Laplacian matrix $L$ is said to be \textbf{k-block-ideal} with respect to a $k$-way partition $\{\mathcal{I}_1,\cdots,\mathcal{I}_k\}$ of $[n]$, if there exists some invertible matrix $S\in\mathbb{R}^{k\times k}$ such that
        \begin{align*}
            V_k:=\bmt v_1(L)&v_2(L)&\cdots &v_k(L)\emt
            =\bmt \one_{\mathcal{I}_1}&\one_{\mathcal{I}_2}&\cdots &\one_{\mathcal{I}_k} \emt S\,.
        \end{align*}
        We also say $V_k$ is k-block-ideal in this case. 
    \end{definition}
    A $k$-block-ideal spectral embedding $V_k$, together with $\Lambda_k$ containing the bottom $k$ eigenvalues of $L$, would immediately lead to a reduced network: the $k$ coherent groups are determined by the $k$-way partition $\{\mathcal{I}_i\}_{i=1}^k$, and the invertible matrix $S$, combined with $\Lambda_k$, characterize the interconnection in the reduced network, as show in the following theorem: 
    \begin{theorem}\label{thm_k_ideal_T_k}
        Given a $k$-block-ideal Laplacian $L$ associated with a partition $\{\mathcal{I}_1,\cdots,\mathcal{I}_k\}$ and an invertible matrix $S$, and we define 
        \begin{align}
            P_{\{\mathcal{I}_i\}_{i=1}^k}:=\bmt \one_{\mathcal{I}_1}&\one_{\mathcal{I}_2}&\cdots &\one_{\mathcal{I}_k} \emt \,,\label{eq_ind_mat}
        \end{align}
        then
        \be
            T_k(s)=P_{\{\mathcal{I}_i\}_{i=1}^k}(I_k+\hat{G}_k(s)L_kf(s))^{-1}\hat{G}_k(s)P^T_{\{\mathcal{I}_i\}_{i=1}^k}\,,
        \ee
        where $\hat{G}(s)=\dg\{\hat{g}_i(s)\}_{i=1}^k\,,\ \hat{g}_i(s)=\lp\sum_{j\in\mathcal{I}_i}g^{-1}_j(s)\rp^{-1}$ and $
            L_k=(S^{-1})^T\Lambda_k S^{-1}$.
    \end{theorem}
    \ifthenelse{\boolean{archive}}{The proof is shown in appendix~\ref{app_pf_thm1_3}.}{We refer the readers to~\citet{min2022b} for the proof.} Theorem \ref{thm_k_ideal_T_k} shows that under $k$-block-ideal $V_k$, the dynamical behavior of $T_k$ is structure-preserving since it is fully characterized by a reduced network with $k$ nodes, with nodal dynamics $\hat{G}(s)$ and network coupling $L_k$. Each node $\hat{g}_i(s)$ represents the aggregate dynamics for nodes in $\mathcal{I}_i$. Any input $u$ to $T_k(s)$ is aggregated into $\bmt \hat{u}_1& \cdots &\hat{u}_k\emt^T=P^T_{\{\mathcal{I}_i\}_{i=1}^k}u$ as the input to the reduced network. Then the output $\bmt \hat{y}_1&\cdots &\hat{y}_k\emt^T$ is ``broadcast" to the original nodes via $P_{\{\mathcal{I}_i\}_{i=1}^k}$ such that every node in the same $\mathcal{I}_i$ has the same response.
    
    Notice that such structure-preserving property only depends on the Laplacian spectral embedding $V_k$. For $V_k$ that is not $k$-block-ideal, we should be able to find a $\hat{V}_k$ close to $V_k$ and is $k$-block-ideal. This gives rise to the following optimization problem:
    \begin{align}
        \min_{ S,\{\mathcal{I}_i\}_{i=1}^k } \|V_k-P_{ \{\mathcal{I}_i\}_{i=1}^k } S\|_F^2,\quad
        s.t.\quad Se_1=\one_k/\sqrt{n},\ S^T\dg\{|\mathcal{I}_i|\}_{i=1}^kS=I_k \label{eq_V_fit_opt}\,.
    \end{align}
    The resulting $\hat{V}_k=P_{ \{\mathcal{I}_i\}_{i=1}^k } S$ is a refinement of $V_k$ that is $k$-ideal, and the constraints in \eqref{eq_V_fit_opt} ensures that the first column of $\hat{V}_k$ is $\one_n/\sqrt{n}$ and that $\hat{V}_k^T\hat{V}_k=I_k$. Now
    $$
        \hat{T}_k(s)=\hat{V}_k(\hat{V}_k^TG^{-1}(s)\hat{V}_k+f(s)\Lambda_k)^{-1}\hat{V}_k^T
    $$
    is structure-preserving by Theorem \ref{thm_k_ideal_T_k}. In the optimization problem \eqref{eq_V_fit_opt}, the need for identifying coherent groups is implicitly suggested by the fact that we are optimizing over all possible $k$-way partitions of $n$, and the reduced network interconnection is constructed by jointly optimizing over invertible $S$. 
    
    Generally, \eqref{eq_V_fit_opt} is hard to solve. Notice, however, that given a fixed partition $\{\mathcal{I}_i\}_{i=1}^k$, one can find a closed-form solution (\ifthenelse{\boolean{archive}}{See Appendix~\ref{app_lse_refine}}{We refer the readers to~\citet{min2022b} for the details}) to the following optimization problem
    \begin{align}
        \min_{ S} \|V_k-P_{ \{\mathcal{I}_i\}_{i=1}^k } S\|_F^2,\quad
        s.t.\quad Se_1=\one_k/\sqrt{n},\ S^T\dg\{|\mathcal{I}_i|\}_{i=1}^kS=I_k \label{eq_V_fit_opt_fix_part}\,.
    \end{align}
    This suggests that a computationally efficient way to find a sub-optimal solution to \eqref{eq_V_fit_opt}: First, we use any spectral clustering algorithm to find a good partition/clustering $\{\mathcal{I}_i\}_{i=1}^k$, then refine the spectral embedding $V_k$ by optimizing \eqref{eq_V_fit_opt_fix_part} with the obtained partition, resulting in our Algorithm \ref{algo_approx_model}. 

\section{Performance Analysis}\label{sec_analysis}
In this section, we provide an error bound on $\sup_{s\in(-j\eta,j\eta)}\|T_{yu}(s)-\hat{T}_k(s)\|$ for our proposed approximation model $\hat{T}_k(s)$ from Algorithm \ref{algo_approx_model}. As we discussed in Section \ref{ssec_low_rank_approx}, such error measure is related to how close the time-domain response of $\hat{T}_k(s)$ is to the one of $T_{yu}(s)$ when subjected to low-frequency inputs. We consider a Laplacian sampled from a stochastic weighted block model. 
\subsection{Weighted Stochastic Block Model}
We first discuss how we sample our Laplacian matrix from a weighted stochastic block model $(\{\mathcal{I}_i\}_{i=1}^k,Q,W)$. Here, $\{\mathcal{I}_i\}_{i=1}^k$ is a $k$-way partition of $[n]$, $Q\in [0,1]^{k\times k}$, and $W\in \mathbb{R}_{\geq 0}^{k\times k}$, where $Q_{ij}=Q_{ji},W_{ij}=W_{ji}$. We let $(j)$ denote the \emph{block membership} of node $j$: when $j\in\mathcal{I}_i$, then $(j)=i$. The adjacency matrix $A$ is sampled as  follows:
\begin{align}
    A_{ij}=\begin{cases}
        W_{(i),(j)}, & \text{with probability } Q_{(i),(j)}\\
        0, &\text{with probability } 1-Q_{(i),(j)}
    \end{cases},\quad i\geq j,\quad\quad A_{ij}=A_{ji},\quad i<j\,.
\end{align}
That is, each (undirected) edge ${i,j}$ appears independently with probability $Q_{(i),(j)}$ that is determined by the block membership of node $i,j$, and has weight $W_{(i),(j)}$ if it appears. Then we have the Laplacian matrix $L$:
\be
    L=D_A-A,\quad D_A=\dg\{A\one\}\,.
\ee
\subsection{Approximation Error Bound}
Given the network model $(G(s),L,f(s))$ with $L$ sampled from a weighted stochastic block model $(\{\mathcal{I}_i\}_{i=1}^k,Q,W)$, we show that under certain assumptions, the error $\sup_{s\in(-j\eta,j\eta)}\|T_{yu}(s)-\hat{T}_k(s)\|$ is small with high probability when the network size is sufficiently large. We start by stating our assumptions.
\begin{assumption}\label{assump_net_model}
    For our network model $(G(s),L,f(s))$ with $L$ sampled from a weighted stochastic block model $(\{\mathcal{I}_i\}_{i=1}^k,Q,W)$, we assume the following:
    \begin{enumerate}[leftmargin=0.5cm]
        \item All $g_i(s),f(s)$ are rational. Moreover, node dynamics are \textbf{output strictly passive}: There exists $\gamma>0$, such that for $i=1,\cdots,n$,
        $
            Re(g_i(s))\geq \frac{1}{\gamma} |g_i(s)|^2, \forall Re(s)>0\,,
        $
        and network coupling $f(s)$ is \textbf{positive real}: $Re(f(s))>0,\forall Re(s)>0$, and $Im(f(s))=0,\forall Re(s)=0$
        \item The node dynamics satisfies that for any $\eta>0$, there exists $M(\eta)$ such that for $i=1,\cdots,n$
        \be
             \sup_{s\in(-j\eta,+j\eta)}|g_i^{-1}(s)|\leq M(\eta)\,.
        \ee
        The network coupling $f(s)$ satisfies that $F_l(\eta):=\inf_{s\in(-j\eta,+j\eta)}|f(s)|$ is positive for all $\eta>0$.
        \item The blocks are approximately balanced: 
        \be
            \frac{n_{\max}}{n_{\min}}\leq \rho\,,\label{eq_assump_blk_bal}
        \ee
        for some $\rho\geq 1$, where $n_{\max}:=\max_{1\leq i\leq k}|\mathcal{I}_i|$ and $n_{\min}:=\min_{1\leq i\leq k}|\mathcal{I}_i|$,

        \item The network has a stronger intra-block connection than the inter-block one: 
        \be
            \min_{i}B_{ii}-2\rho \max_{i}\sum_{j\neq i}B_{ij}\geq \Delta\,,\label{eq_assump_con_diff}
        \ee
        for some $\Delta>0$, where $B=Q\odot W$. ($\odot$ is the Hadamard product)
        \end{enumerate}
\end{assumption}
The first assumption ensures the network $T_{yu}(s)$ and our approximation model $\hat{T}_k(s)$ are stable. The second assumption ensures that our low-rank approximation $T_k(s)$ in Theorem \ref{thm_Tk} is valid on the interval of our interest $(-j\eta, +j\eta)$. The third assumption ensures our problem is non-degenerate: if the size of one block is too small, the network effectively has $k-1$ clusters. Such an assumption is standard in analyzing the consistency of spectral clustering algorithms on stochastic block models~\citep{lyzinski2014perfect,ahn2018hypergraph}. Lastly, since we are interested in networks  containing multiple groups of nodes such that within each group, the nodes are tightly-connected while between groups, the nodes are weakly-connected, the fourth assumption formally characterizes such a property.

In our algorithm, a spectral clustering algorithm is used to find a partition $\{\mathcal{I}_i\}_{i=1}^k$ that is used for aggregating node dynamics and constructing the reduced network. Ideally, we want some consistency property on the obtained partition.

\begin{assumption}\label{assump_sc_algo}
    Given $L$ sampled from a weighted stochastic block model $(\{\mathcal{I}_i\}_{i=1}^k,Q,W)$ satisfying Assumption \ref{assump_net_model}, we have an asymptotically consistent spectral clustering algorithm in Algorithm \ref{algo_approx_model}: For any $\delta>0$, there exists $\tilde{N}(\delta)$ such that for network with size $n>\tilde{N}(\delta)$, the spectral clustering algorithm on $L$ returns the true $\{\mathcal{I}_i\}_{i=1}^k$ partition with probability at least $1-\delta$.
\end{assumption}
Formally justifying this assumption for some spectral clustering algorithms is an interesting future research topic. Nonetheless, such a consistency result has been studied for spectral clustering algorithms on the adjacency matrix from the stochastic block model~\citep{lyzinski2014perfect} and the weighted stochastic block model~\citep{ahn2018hypergraph}. 

With these assumptions, we have the following theorem regarding the error bound.
\begin{theorem}\label{thm_err}
    Consider the network model $(G(s),L,f(s))$ with $L$ sampled from a weighted stochastic block model $(\{\mathcal{I}_i\}_{i=1}^k,Q,W)$. If Assumption \ref{assump_net_model} and Assumption \ref{assump_sc_algo} hold, then Algorithm \ref{algo_approx_model} returns a $\hat{T}_k(s)$ such that 
    \begin{enumerate}[leftmargin=0.5cm]
        \item $\|T_{yu}(s)\|_{\mathcal{H}_\infty}\leq \gamma$, $\|\hat{T}_k(s)\|_{\mathcal{H}_\infty}\leq \gamma$;
        \item For any $\eta>0$, $0<\delta<1$ and $\epsilon>0$, there exists an $N(\delta,\epsilon,\tilde{N}(\delta/2),\gamma,M(\eta), F_l(\eta),\rho,Q,W)$ such that for network with size $n\geq N$, with probability at least $1-\delta$, we have
        \be
            \sup_{s\in(-j\eta,+j\eta)}\|T_{yu}(s)-\hat{T}_k(s)\|\leq \epsilon\,.
        \ee
    \end{enumerate}
\end{theorem}
\begin{proof}[Proof Sketch]
    For the stability of $T_{yu}(s)$ and $\hat{T}_k(s)$, the proof is similar to the one in~\citet{min2021a} and uses the assumption $g_i(s)$ are output strictly passive and $f(s)$ is positive real. The error bound relies on that the sampled Laplacian matrix $L$ is close to one that is easy to analyze: Let $A_{\text{blk}}$ be the expected value of the adjacency matrix $A$ from the block model, and we can construct a Laplacian matrix $L_{\text{blk}}=D_{A_{\text{blk}}}-A_{\text{blk}},\ D_{A_\text{blk}}=\dg\{A_{\text{blk}}\one\}$. $L_\text{blk}$ has (by \eqref{eq_assump_blk_bal}\eqref{eq_assump_con_diff} in Assumption \ref{assump_net_model}) all the desired properties: 1) $\lambda_{k+1}(L_{\text{blk}})$ grows linearly in network size $n$; 2) $L_\text{blk}$ is $k$-block-ideal.~\citet{min2022} has shown that under the weighted stochastic block model, $\|L-L_{\text{blk}}\|\sim\mathcal{O}_p(\sqrt{n\log n})$, which is sufficient to show that 1) $\lambda_{k+1}(L)\sim \Omega_p(n)$ by Weyl's inequality~\citep{Horn:2012:MA:2422911}; 2) $L$ is approximately $k$-block-ideal by Davis-Khan theorem~\citep{Yu2014}. The former shows that the error between $T_{yu}(s)$ and $T_k(s)$ is small w.h.p. by Theorem \ref{algo_approx_model} and the latter ensures the error between $T_k(s)$ and $\hat{T}_k(s)$ is small w.h.p. \ifthenelse{\boolean{archive}}{}{We refer the readers to~\citet{min2022b} for the full proof.}
\end{proof}
\ifthenelse{\boolean{archive}}{The proof is shown in Appendix~\ref{app_pf_thm_err}. }{}
Theorem \ref{thm_err} shows that our algorithms perform well for large networks with multiple coherent clusters, it also implies that the collective dynamic behavior of such networks can be modeled as a structured reduced network. This suggests a new avenue for data-driven system identification for such networks where only the reduced network model is learned from the data collected from the network.
\section{Numerical Experiments}\label{sec_num}
The frequency response of synchronous generator (including grid-forming inverters) networks, linearized at its equilibrium point~\citep{zhao2013power}, can be modeled exactly as the network model in Fig \ref{fig_blk_digm} with $f(s)=\frac{1}{s}$ and second order node dynamics $g_i(s)$. 
We validate our algorithm with a synthetic test case, where the coefficients of generator dynamics are randomly sampled.
The network adjacency matrix $A$ is sampled from our weighted stochastic block model $(\{\mathcal{I}_i\}_{i=1}^k,Q,W)$ with $k=3$, and
\be
    \bmt |\mathcal{I}_1| & 0 &0\\
     0 & |\mathcal{I}_2| &0\\
     0 & 0 & |\mathcal{I}_3|\emt=\bmt 20 & 0 & 0\\ 0& 40 &0 \\ 0 & 0& 20\emt, Q=\bmt 0.8 & 0.1 & 0.1\\ 0.1& 0.8 &0.1 \\ 0.1 & 0.1& 0.8\emt, W=\bmt 20 & 0.4 & 0.8\\ 0.4& 20 &0.7 \\ 0.8 & 0.7& 20\emt\,.
\ee
We use the spectral clustering algorithm proposed in~\citet{bach2004learning}. Since the network size is not sufficiently large for the algorithm to return a true partition with high probability, when we run the experiments with multiple random seeds, we see a small fraction of the runs in which the algorithm fails to cover the true partition. For the case when the spectral clustering algorithm succeeds, we inject a step disturbance $u_2(t)=\chi(t)$ at the second node of the network and plot the step response of $T_{yu}(s)$ in Fig \ref{fig_syn}, along with the response $\hat{y}$ of our approximate model $\hat{T}_3(s)$ from Algorithm \ref{algo_approx_model}. There is a clear difference between the dynamical response of generators from different groups, and the aggregate responses $\hat{y}$ capture such difference while providing a good approximation to the actual node responses.
Due to space constraints, we only present the result of running Algorithm \ref{algo_approx_model} on one instance of the randomly generated networks, but the results are consistent across multiple runs as long as the spectral clustering succeeds. 
\begin{figure}[!h]
    \centering
    \includegraphics[width=\columnwidth]{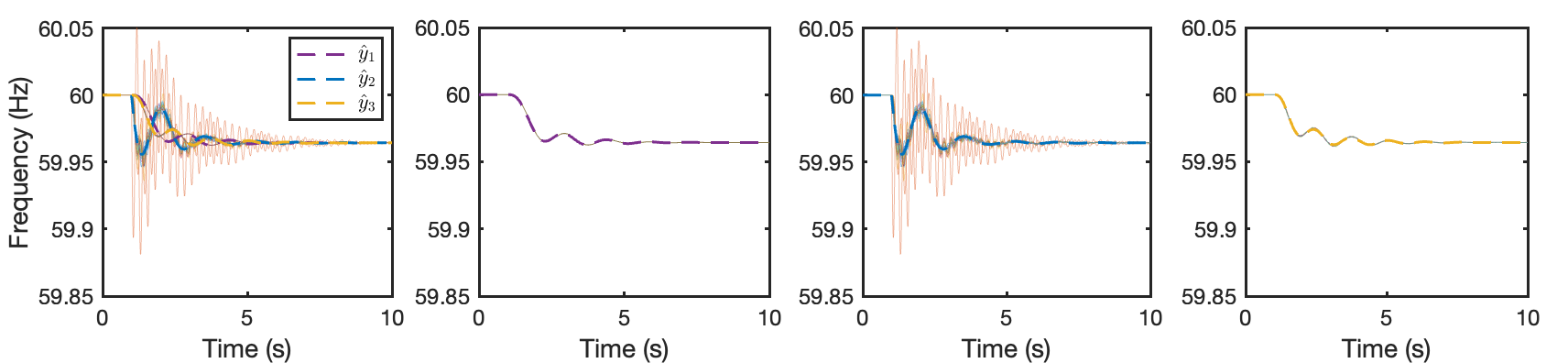}
    \caption{Left most plot shows the step response of $T_{yu}(s)$ (solid lines) and $\hat{T}_k(s)$ (dashed lines) from algorithm \ref{algo_approx_model}. The three plots on the right show the response for each identified group $\mathcal{I}_{i}$. The node injected with step disturbance is in the group $2$.}
    \label{fig_syn}
\end{figure}

\section{Conclusion}
In this paper, we propose a structure-preserving model-reduction methodology for large-scale dynamic networks based on a recent frequency-domain characterization of coherent dynamics in networked systems. Our analysis shows that networks with multiple coherent groups can be well approximated by a reduced network of the same size as the number of coherent groups, and we provide an upper bound on the approximation error when the network graph is randomly generated from a weight stochastic block model. We believe our proposed model can be applied to power networks for studying the inter-area oscillation in the frequency response and allows new control designs based on the reduced network.
\acks{This work was supported by the NSF HDR TRIPODS Institute for the Foundations of Graph and Deep Learning (NSF grant 1934979), the ONR MURI Program (ONR grant 503405-78051), the NSF CPS Program (NSF grant 2136324), and the NSF CAREER Program (NSF grant 1752362). The authors thank Professor Steven Low for the insightful discussion that inspired this work. The authors also thank Bala Kameshwar Poolla and Dhananjay Anand for carefully reading our manuscript and for providing many insightful comments and suggestions.}

\bibliography{ref}

\ifthenelse{\boolean{archive}}{\input{appendix}}{}
\end{document}

%% file: edits.tex
\usepackage{ifthen}
\newboolean{showcomments}
\setboolean{showcomments}{true}
\usepackage{todonotes}

\definecolor{bleudefrance}{rgb}{0.19, 0.55, 0.91}
\definecolor{ao(english)}{rgb}{0.0, 0.5, 0.0}

\newcommand{\addcite}[0]{\ifthenelse{\boolean{showcomments}}
{\textcolor{purple}{(add cite(s)) }}{}}%

\newcommand{\hl}[1]{\ifthenelse{\boolean{showcomments}}
{\textcolor{red}{#1}}{#1}}
\newcommand{\enrique}[1]{  \ifthenelse{\boolean{showcomments}}
{\todo[inline,color=bleudefrance]{Enrique: #1}}{}}
\newcommand{\emmargin}[1]{\ifthenelse{\boolean{showcomments}}{\marginpar{\color{bleudefrance}\tiny EM: #1}}{}}
\newcommand{\hancheng}[1]{  \ifthenelse{\boolean{showcomments}}
{\todo[inline,color=orange]{Hancheng: #1}}{}}


%% file: appendix.tex
\newpage
\appendix

\section{Solution to the Laplacian spectral embedding refinement}\label{app_lse_refine}
In this section, we derive the analytical solution to \eqref{eq_V_fit_opt_fix_part}:
\begin{align*}
        \min_{ S\in\mathbb{R}^{k\times k}}&\;\quad \|V_k-P_{ \{\mathcal{I}_i\}_{i=1}^k } S\|_F^2\\
        s.t.\quad &\; Se_1=\one_k/\sqrt{n}\nonumber\\
        &\; S^T\dg\{|\mathcal{I}_i|\}_{i=1}^kS=I_k\nonumber\,.
    \end{align*}
First of all, there is nothing to optimize in the first column of $S$, since $S$ must be of the form $S=\bmt \frac{\one_k}{\sqrt{n}} & \tilde{S}\emt$ for some $\tilde{S}\in\mathbb{R}^{k\times (k-1)}$. Since $V_k=\bmt \frac{\one_n}{\sqrt{n}}& \tilde{V}_k\emt$ with $\tilde{V}_k=\bmt v_2(L)& \cdots & v_k(L)\emt$, solving \eqref{eq_V_fit_opt_fix_part} is equivalent to solving 
\begin{align}
    \min_{ \tilde{S}\in\mathbb{R}^{k\times (k-1)}}&\;\quad \|\tilde{V}_k-P_{ \{\mathcal{I}_i\}_{i=1}^k } \tilde{S}\|_F^2\label{eq_V_fit_opt_fix_part_equiv1}\\
    s.t.\quad &\; \tilde{S}^T\one_k=0\nonumber\\
    &\; \tilde{S}^T\dg\{|\mathcal{I}_i|\}_{i=1}^k\tilde{S}=I_k\nonumber\,,
\end{align}
where the first constraint in \eqref{eq_V_fit_opt_fix_part} is removed by excluding the first column of $S$, and the second constraint in \eqref{eq_V_fit_opt_fix_part} is rewritten as the two constraints in \eqref{eq_V_fit_opt_fix_part_equiv1}.

Now define $\tilde{O}:=\dg\{\sqrt{|\mathcal{I}_i|}\}_{i=1}^k\tilde{S}\in\mathbb{R}^{k\times(k-1)}$ and $\tilde{P}_{ \{\mathcal{I}_i\}_{i=1}^k }=P_{ \{\mathcal{I}_i\}_{i=1}^k} \dg\{(\sqrt{|\mathcal{I}_i|})^{-1}\}_{i=1}^k$, it is easy to see that \eqref{eq_V_fit_opt_fix_part_equiv1} is equivalent to:
\begin{align}
    \min_{ \tilde{O}\in\mathbb{R}^{k\times (k-1)}}&\;\quad \|\tilde{V}_k-\tilde{P}_{ \{\mathcal{I}_i\}_{i=1}^k } \tilde{O}\|_F^2\label{eq_V_fit_opt_fix_part_equiv2}\\
    s.t.\quad &\; \tilde{O}^Tu_{\{\mathcal{I}_i\}_{i=1}^k}=0\nonumber\\
    &\; \tilde{O}^T\tilde{O}=I_k\nonumber\,,
\end{align}
where $u_{\{\mathcal{I}_i\}_{i=1}^k}=\dg\{(\sqrt{|\mathcal{I}_i|})^{-1}\}_{i=1}^k\one_k$. Now let $Q\in\mathbb{R}^{k\times (k-1)}$ be some matrix such that $Q^TQ=I$ and $QQ^T=I-u_{\{\mathcal{I}_i\}_{i=1}^k}u^T_{\{\mathcal{I}_i\}_{i=1}^k}$, then $\{QO:\ O\in\mathbb{R}^{(k-1)\times (k-1)}, O^TO=OO^T=I_{k-1}$ are all the feasible solution to \eqref{eq_V_fit_opt_fix_part_equiv2}. Therefore \eqref{eq_V_fit_opt_fix_part_equiv2} is equivalent to
\begin{align}
    \min_{ O\in\mathbb{R}^{(k-1)\times (k-1)}}&\;\quad \|\tilde{V}_k-\tilde{P}_{ \{\mathcal{I}_i\}_{i=1}^k } QO\|_F^2\label{eq_V_fit_opt_fix_part_equiv3}\\
    s.t.\quad 
    &\; O^TO=I_{k-1}\nonumber\,.
\end{align}
Given the SVD: $Q^T\tilde{P}^T_{ \{\mathcal{I}_i\}_{i=1}^k }\tilde{V}_k=U\Sigma V^T$, the optimal solution to \eqref{eq_V_fit_opt_fix_part_equiv3} is $O^*=UV^T$. Then the optimal solution to the original problem \eqref{eq_V_fit_opt_fix_part} is 
\be
    S^*=\bmt \frac{\one_{k}}{\sqrt{n}}& \dg\{(\sqrt{|\mathcal{I}_i|})^{-1}\}_{i=1}^k QO^*\emt\,.
\ee

\newpage

\section{Proof of Theorem \ref{thm_Tk} and Theorem \ref{thm_k_ideal_T_k}}\label{app_pf_thm1_3}
\begin{proof}[Proof of Theorem \ref{thm_Tk}]
    Firstly, we have
    \begin{align}
        T_{yu}(s_0)=(I_n+G(s_0)f(s_0)L)^{-1}G(s_0)&=\; (G^{-1}(s_0)+f(s_0)L)^{-1}\\
        &=\; V(V^TG^{-1}(s_0)V+f(s_0)\Lambda )^{-1}V^T\,,
    \end{align}
    where $G^{-1}(s_0)=\dg\{g^{-1}_i(s_0)\}$, $\Lambda=\dg\{\lambda_i(L)\}$, and $V=\bmt v_1(L),V_k(L),\cdots,v_n(L)\emt$.
    
    Let $H=V^T\dg\{g_i^{-1}(s_0)\}V+f(s_0)\Lambda$, then 
    \ben
    T_{yu}(s_0) = VH^{-1}V^T.
    \een
    Then it is easy to see that
    \begin{align}
        \lV T_{yu}(s_0)-T_k(s_0)\rV= \lV T(s_0)-V\bmt H_k^{-1} &0 \\
        0& 0\emt V^T\rV
        &=\; \lV V\lp H^{-1}-\bmt H_k^{-1} &0 \\
        0& 0\emt\rp V^T\rV\nonumber\\
        &=\; \lV H^{-1}-\bmt H_k^{-1} &0 \\
        0& 0\emt \rV\,,\label{eq_T_H_norm_equiv}
    \end{align}
    where the last equality comes from the fact that multiplying by a unitary matrix $V$ preserves the spectral norm.
    
    Let $V_k:=\bmt\frac{\one}{\sqrt{n}} & V_k(L),\cdots,v_k(L)\emt$ and $ V_k^\perp:=\bmt v_{k+1}(L)& \cdots &v_n(L)\emt$, we now write $H$ in block matrix form:
    \begin{align*}
        H&=\;V^T\dg\{g_i^{-1}(s_0)\}V+f(s_0)\Lambda\\
        &= \bmt
            V_k^T\\
            (V_k^\perp)^T
        \emt \dg\{g_i^{-1}(s_0)\} \bmt
        V_k& V_k^\perp\emt+f(s_0)\Lambda\\
        &= {\small\bmt
        H_2& V_k^T\dg\{g_i^{-1}(s_0)\}V_k^\perp\\
        (V_k^\perp)^T\dg\{g_i^{-1}(s_0)\}V_k &(V_k^\perp)^T\dg\{g_i^{-1}(s_0)\}V_k^\perp+f(s_0)\Tilde{\Lambda}
        \emt}
        := \bmt
        H_k& H^T_o\\
        H_o & H_d
        \emt\,,
    \end{align*}
    where  $\Tilde{\Lambda}=\dg\{\lambda_{k+1}(L),\cdots,\lambda_n(L)\}$.
    
    Inverting $H$ in its block form, we have
    \ben
        H^{-1} = \bmt
        A &-A H_o^TH_d^{-1}\\
        -H_d^{-1}H_o A& H_d^{-1}+H_d^{-1}H_o A H_o^TH_d^{-1}
        \emt\,,
    \een
    where $A = (H_k-H_o^TH_d^{-1}H_o)^{-1}$.
    
    Notice that $||V_k^\perp||=1$ and $||V_k||=1$, we have
    \begin{align}
        \|H_o\|&=\; \lV (V_k\perp)^T\dg\{g_i^{-1}(s_0)\}V_k\rV\nonumber\\
        &\leq\; \|V_\perp\|\|\dg\{g_i^{-1}(s_0)\}\|\|V_k\|\leq M_2\,.\label{eq_h12_norm_bd}
    \end{align}

    Also, by Weyl's inequality~\citep{Horn:2012:MA:2422911}, when $|f(s_0)|\lambda_{k+1}(L)>M_2$, the following holds:
    \begin{align}
        \|H_d^{-1}\|&=\;\|(f(s_0)\Tilde{\Lambda}+ (V_k^\perp)^T\dg\{g_i^{-1}(s_0)\}V_k^\perp)^{-1}\|\nonumber\\
        &\leq\; \frac{1}{\sigma_{\min}(f(s_0)\Tilde{\Lambda})-\|(V_k^\perp)^T\dg\{g_i^{-1}(s_0)\}V_k^\perp\|}\nonumber\\
        &\leq\; \frac{1}{\sigma_{\min}(f(s_0)\Tilde{\Lambda})-M_2}\leq \frac{1}{|f(s_0)|\lambda_{k+1}(L)-M_2} \,.\label{eq_H22_norm_bd}
    \end{align}
    Lastly, when $|f(s_0)|\lambda_{k+1}(L)>M_2+M_2^2M_1$, a similar reasoning as above, using \eqref{eq_h12_norm_bd} \eqref{eq_H22_norm_bd}, and our assumption $\|T_k(s_0)\|=\|H_k^{-1}\|\leq M_1$, gives
    \begin{align}
        \|A\|\leq \frac{1}{\|H_k\|-\|H_o^TH_d^{-1}H_o\|}&\leq \frac{1}{\|H_k\|-\|H_o\|^2\|H_d^{-1}\|}\nonumber\\
        &\leq\; \frac{1}{\frac{1}{M_1}-\frac{M_2^2}{|f(s_0)|\lambda_{k+1}(L)-M_2}}\nonumber\\
        &= \;
        \frac{(|f(s_0)|\lambda_{k+1}(L)-M_2)M_1}{|f(s_0)|\lambda_{k+1}(L)-M_2-M_1M_2^2}\,.\label{eq_a_norm_bd}
    \end{align}
    
    Now we bound the norm of $H^{-1}-\bmt H_k^{-1} &0 \\
        0& 0\emt$ by the sum of norms of all its blocks:
    \begin{align}
        \lV H^{-1}-\bmt H_k^{-1} &0 \\
        0& 0\emt \rV
        =&\; \lV \bmt
        AH_o^TH_d^{-1}H_o H_2 &-AH_o^TH_d^{-1}\\
        -aH_d^{-1}H_o& H_d^{-1}+AH_d^{-1}H_oH_o^TH_d^{-1}
        \emt\rV\nonumber\\
        \leq &\; \|AH_o^TH_d^{-1}H_o H_2\|+2\|AH_d^{-1}H_o\|\nonumber\\
        &\; +\|H_d^{-1}+AH_d^{-1}H_oH_o^TH_d^{-1}\|\nonumber\\
        \leq &\; \|A\|\|H_d^{-1}\|(\|H_2\|\|H_o\|^2+2\|H_o\|+\|H_o\|^2\|H_d^{-1}\|) +\|H_d^{-1}\|\,,\label{eq_Hinv_norm_bd1}
    \end{align}
    Using \eqref{eq_h12_norm_bd}\eqref{eq_H22_norm_bd}\eqref{eq_a_norm_bd}, we can further upper bound \eqref{eq_Hinv_norm_bd1} as
    \begin{align}
        \lV H^{-1}-\bmt H_k^{-1} &0 \\
        0& 0\emt \rV
        \leq &\;  \frac{M_1^2M_2^2+2M_1M_2+\frac{M_1M_2^2}{|f(s_0)|\lambda_{k+1}(L)-M_2}}{|f(s_0)|\lambda_{k+1}(L)-M_2-M_1M_2^2}+\frac{1}{|f(s_0)|\lambda_{k+1}(L)-M_2}\nonumber\\
        =&\;\frac{\lp M_1M_2+1\rp^2}{|f(s_0)|\lambda_{k+1}(L)-M_2-M_1M_2^2}\,.\label{eq_Hinv_norm_bd2}
    \end{align}
    This bound holds as long as $|f(s_0)|\lambda_{k+1}(L)>M_2+M_2^2M_1$. Combining \eqref{eq_T_H_norm_equiv} and \eqref{eq_Hinv_norm_bd2} gives the desired inequality.
\end{proof}
\begin{proof}[Proof of Theorem \ref{thm_k_ideal_T_k}]
    Since
    \be
        T_k(s)=V_k(V_k^TG^{-1}(s)V_k+f(s)\Lambda_k)^{-1}V_k^T\,,
    \ee
    and
    \be
        V_k=P_{\{\mathcal{I}_i\}_{i=1}^k}S\,,
    \ee
    we have
    \begin{align*}
        T_k(s)&=\;P_{\{\mathcal{I}_i\}_{i=1}^k}S\lp S^TP^T_{\{\mathcal{I}_i\}_{i=1}^k}G^{-1}(s)P_{\{\mathcal{I}_i\}_{i=1}^k}S+f(s)\Lambda_k\rp^{-1}S^TP^T_{\{\mathcal{I}_i\}_{i=1}^k}\\
        &=\;P_{\{\mathcal{I}_i\}_{i=1}^k}\lp (S^T)^{-1}\lp S^TP^T_{\{\mathcal{I}_i\}_{i=1}^k}G^{-1}(s)P_{\{\mathcal{I}_i\}_{i=1}^k}S+f(s)\Lambda_k\rp S^{-1}\rp^{-1}P^T_{\{\mathcal{I}_i\}_{i=1}^k}\\
        &=\;P_{\{\mathcal{I}_i\}_{i=1}^k}\lp P^T_{\{\mathcal{I}_i\}_{i=1}^k}G^{-1}(s)P_{\{\mathcal{I}_i\}_{i=1}^k}+f(s)(S^T)^{-1}\Lambda_kS^{-1}\rp^{-1}P^T_{\{\mathcal{I}_i\}_{i=1}^k}\\
        &=\;P_{\{\mathcal{I}_i\}_{i=1}^k}\lp \hat{G}^{-1}_k(s)+f(s)L_k\rp^{-1}P^T_{\{\mathcal{I}_i\}_{i=1}^k}\\
        &=\; P_{\{\mathcal{I}_i\}_{i=1}^k}\lp I+\hat{G}(s)L_kf(s)\rp^{-1}\hat{G}(s)P^T_{\{\mathcal{I}_i\}_{i=1}^k}\,.
    \end{align*}
        
\end{proof}
\newpage
\section{Proof of Theorem \ref{thm_err}}\label{app_pf_thm_err}
Before showing the proof of Theorem \ref{thm_err}, we state a few lemmas that are used. 
\subsection{Auxiliary Lemmas}
The proofs of the lemmas shown in this section are presented in Appendix~\ref{app_ssec_pf_ax_lem}.

Firstly, we need the following lemma concerning the stability of the original network $T_{yu}(s)$ and our approximation model $T_k(s),\hat{T}_k(s)$. 
\begin{lemma}
    \label{lem_Tk_hinf_bd}
    Suppose all $g_i(s), f(s)$ satisfies Assumption \ref{assump_net_model}, then for any $V_k$ with $V_k^TV_k=I$ and any $\Lambda_k\succeq 0$, the corresponding $$T(s)=V_k(V_k^T\dg\{g_i^{-1}(s)\}V_k+f(s)\Lambda_k)V_k^T$$
    has
    $$
      \|T(s)\|_{\mathcal{H}_\infty}\leq \gamma\,.
    $$
\end{lemma}
This Lemma shows the stability of $T_{yu}(s),T_k(s),\hat{T}_k(s)$ by choosing different $V_k,\Lambda_k$.

The following lemma concerns controlling the approximation error between $T_k(s)$ and $\hat{T}_k(s)$.
\begin{lemma}\label{lem_T_k_approx_err}
    Suppose all $g_i(s), f(s)$ satisfies Assumption \ref{assump_net_model}. Given two matrices $V_k,\hat{V_k}\in\mathbb{R}^{n\times k}$ with $V_k^TV_k=\hat{V}_k^T\hat{V}_k=I$ and some $\Lambda_k\succeq 0$. Define
    \begin{align*}
        T_k(s)&=\;V_k(V_k^T\dg\{g_i^{-1}(s)\}V_k+f(s)\Lambda_k)^{-1}V_k^T\,,\\
        \hat{T}_k(s)&=\;\hat{V}_k(\hat{V}_k^T\dg\{g_i^{-1}(s)\}\hat{V}_k+f(s)\Lambda_k)^{-1}\hat{V}_k^T\,.
    \end{align*}
    Given any $\eta>0$, we have
    $$\sup_{s\in(-j\eta,+j\eta)}\|T_k(s)-\hat{T}_k(s)\|\leq 2(\gamma+\gamma^2M(\eta)) \|V_k-\hat{V}_k\|_F\,,$$
    where $
        M(\eta)=\sup_{s\in(-j\eta,+j\eta)}\max_i|g_i^{-1}(s)|
    $.
\end{lemma}
That is, since $\hat{T}_k(s)$ is obtained by replace $V_k$ in $T_k(s)$ by $\hat{T}_k(s)$,  the error can be controlled by the difference $\|V_k-\hat{V}_k\|_F$ between $V_k$ and $\hat{V}_k$. Recall that Theorem \ref{thm_Tk} provides a bound on $\|T_{yu}(s)-T_k(s)\|$, combing it with Lemma \ref{lem_T_k_approx_err} allows us to control the error $\|T_{yu}(s)-\hat{T}_k(s)\|$, as stated in the following lemma:
\begin{lemma}\label{lem_err_prob_bd}
    Consider the network model $(G(s),L,f(s))$ with $L$ sampled from a weighted stochastic block model $(\{\mathcal{I}_i\}_{i=1}^k,P,W)$. If Assumption \ref{assump_net_model} holds,
    Then given any $\eta>0$,
    we have $\forall\epsilon>0$,
    \begin{align*}
        &\;\prob\lp \sup_{s\in(-j\eta,+j\eta)}\|T_{yu}(s)-\hat{T}_k(s)\|\geq \epsilon\rp\\
        \leq &\;2\prob\lp \lambda_{k+1}(L)\leq \frac{1}{F_l(\eta)}\lp \frac{2}{\epsilon}(\gamma M(\eta) +1)^2+M(\eta)+\gamma M^2(\eta)\rp\rp\\
        &\; \qquad\qquad\qquad\qquad\qquad+\prob\lp \|V_k-\hat{V}_k\|\geq \frac{\epsilon}{4(\gamma+\gamma^2M(\eta))}\rp\,,
    \end{align*}
    where $
        \sup_{s\in(-j\eta,+j\eta)}\max_i|g_i^{-1}(s)|:=M(\eta)
    $ and $F_l(\eta):=\inf_{s\in(-j\eta,+j\eta)}|f(s)|$.
\end{lemma}
That is, we need to lower bound $\lambda_{k+1}(L)$ and upper bound $\|V_k-\hat{V}_k\|$ for controlling the error. All of these are possible by studying the Laplacian matrix constructed from the expected adjacency matrix:

For a weighted stochastic block model $(\{\mathcal{I}_i\}_{i=1}^k,Q,W)$, we denote the expected value of adjacency matrix $A$ as
\be
    A_{\text{blk}}=P_{\{\mathcal{I}_i\}_{i=1}^k}BP_{\{\mathcal{I}_i\}_{i=1}^k}^T,\ B=Q\odot W\,,
\ee
and define
\be
    L_\text{blk}=\dg\{A_\text{blk}\one_n\}-A_\text{blk}\,,
\ee
and
\be
    V_k^\text{blk}=\bmt \frac{\one}{\sqrt{n}}& v_2(L_\text{blk}) &\cdots & v_k(L_\text{blk})\emt\,.
\ee
Firstly, if the spectral clustering algorithm returns the true block assignment, then it is sufficient to control the difference between $V_k$ and $V_k^\text{blk}$ for upper bounding $\|V_k-\hat{V}_k\|$:
\begin{lemma}\label{lem_fro_bd_sin}
    Let $V_k^\text{blk}=\bmt \frac{\one}{\sqrt{n}}& v_2(L_\text{blk}) &\cdots & v_k(L_\text{blk})\emt$. Consider an $L$ sampled from a weighted stochastic block model $(\{\mathcal{I}_i\}_{i=1}^k,Q,W)$, if the spectral clustering algorithm in Algorithm \ref{algo_approx_model} returns the true block assignments $\{\mathcal{I}_i\}_{i=1}^k$, then optimizing \eqref{eq_V_fit_opt_fix_part} yields a $\hat{V}$ such that
    \be
        \|\hat{V}_k-V_k\|\leq \|\sin \Theta(V_k,V_k^{\text{blk}})\|_F\,.
    \ee
\end{lemma}
The term $\|\sin \Theta(V_k,V_k^{\text{blk}})\|_F$ should be small given that $L$ and $L_{blk}$ are sufficiently close to each other with high probability(to be formalized later). Moreover, $\lambda_{k+1}(L)$ and $\lambda_{k+1}(L_{\text{blk}})$ should be close for the same reason. We discuss the spectrum of $L_\text{blk}$ in detail in Appendix~\ref{app_spec_L_blk}. The following Lemma is the direct consequence of Proposition \ref{prop_eigenspace_L_blk} in Appendix~\ref{app_spec_L_blk}:
\begin{lemma}\label{lem_spec_L_blk}
    Consider a weighted stochastic block model $(\{\mathcal{I}_i\}_{i=1}^k,Q,W)$ satisfying Assumption \ref{assump_net_model}. Let $n_{\min}=\min_{1\leq i\leq k}|\mathcal{I}_i|$, and $b_{\min}:=\min\{[B_k\one_k]_i:\ i=1,\cdots,k\}$. We have
    \begin{enumerate}
        \item $L_\text{blk}$ is $k$-block-ideal;
        \item $\lambda_{k+1}(L_{\text{blk}})\geq b_{\min}n_{\min}$;
        \item $\lambda_{k+1}(L_{\text{blk}})-\lambda_{k}(L_{\text{blk}})\geq \Delta n_{\min}$.
    \end{enumerate}
\end{lemma}
Now we are ready to proof our main theorem.
\subsection{Proof of Theorem \ref{thm_err}}

\begin{proof}[Proof of Theorem \ref{thm_err}]
    Define
    $\tilde{B}:=P\odot W \odot W$, and $\tilde{b}_{\max}=\max_i\sum_{j}B_{ij},\tilde{b}_{\min}=\min_j\sum_{j}B_{ij}$. We also define $W_{\max}=:\max_{ij}|W_{ij}|$.
    A direct application of Proposition 3 in~\citet{min2022} shows that for any $c>0$, if 
    \be kn_{\min}\tilde{b}_{\max}\geq 16(c+1)\log n\,,\label{eq_pf_thm_err_cond1}\ee
    then for any $4n^{-c}\leq \frac{\delta}{6}<1$, we have
    \be
    \prob\lp \|L-L_{\text{blk}}\|\geq 8\sqrt{kn_{\max}\tilde{b}_{\max}\log (24n/\delta)}\rp\leq \frac{\delta}{6}
    \ee
    If 
    \be
        b_{\min}n_{\min}-8\sqrt{kn_{\max}\tilde{b}_{\max}\log (24n/\delta)}\geq \frac{1}{F_l(\eta)}\lp \frac{2}{\epsilon}(\gamma M(\eta) +1)^2+M(\eta)+\gamma M^2(\eta)\rp\label{eq_pf_thm_err_cond2}\,,
    \ee
    then 
    \begin{align*}
        &\;\lambda_{k+1}(L)\leq \frac{1}{F_l(\eta)}\lp \frac{2}{\epsilon}(\gamma M(\eta) +1)^2+M(\eta)+\gamma M^2(\eta)\rp\\
        \Ra &\;\lambda_{k+1}(L)\leq b_{\min}n_{\min}-8\sqrt{kn_{\max}\tilde{b}_{\max}\log (24n/\delta)}\\
         (\text{Lemma \ref{lem_spec_L_blk}})&\;\\
        \Ra &\; \lambda_{k+1}(L)\leq \lambda_{k+1}(L_{\text{blk}})-8\sqrt{kn_{\max}\tilde{b}_{\max}\log (24n/\delta)}\\
        \Ra &\;8\sqrt{kn_{\max}\tilde{b}_{\max}\log (24n/\delta)}\leq \lambda_{k+1}(L_{\text{blk}})-\lambda_{k+1}(L)\\
        &\;(\text{Weyl's inequality~\citep{Horn:2012:MA:2422911}})\\
        \Ra &\; 8\sqrt{kn_{\max}\tilde{b}_{\max}\log (24n/\delta)}\leq\|L-L_{\text{blk}}\|\,.
    \end{align*}
    That is, for a given $\delta\geq 4n^{-c}$, if \eqref{eq_pf_thm_err_cond1}\eqref{eq_pf_thm_err_cond2} hold, then
    \begin{align*}
        &\;\prob\lp \lambda_{k+1}(L)\leq \frac{1}{F_l(\eta)}\lp \frac{2}{\epsilon}(\gamma M(\eta) +1)^2+M(\eta)+\gamma M^2(\eta)\rp\rp\\
        &\;\quad\quad\quad\quad\leq \prob\lp \|L-L_{\text{blk}}\|\geq 8\sqrt{kn_{\max}\tilde{b}_{\max}\log (24n/\delta)}\rp\leq \frac{\delta}{6}\,.
    \end{align*}
    Similarly, when
    \be
        \frac{\epsilon}{8\sqrt{k}(\gamma+\gamma^2M(\eta))}\Delta n_{\min}\geq 8\sqrt{kn_{\max}\tilde{b}_{\max}\log (24n/\delta)}\label{eq_pf_thm_err_cond3}\,,
    \ee
    and the spectral clustering (SC) returns the true $\{\mathcal{I}_i\}_{i=1}^k$, then
    \begin{align*}
        &\;\|V_k-\hat{V}_k\|\geq \frac{\epsilon}{4(\gamma+\gamma^2M(\eta))}\\
        (\text{Lemma \ref{lem_fro_bd_sin}}) &\;\\
        \Ra &\; \|\sin \Theta(V_k,V_k^{\text{blk}})\|_F\geq \frac{\epsilon}{4(\gamma+\gamma^2M(\eta))}\\
        (\text{Davis-Khan~\citep{Yu2014}})&\;\\
        \Ra &\; \frac{2\sqrt{k}\|L-L_{\text{blk}}\|}{\lambda_{k+1}(L_{\text{blk}})-\lambda_k(L_{\text{blk}})}\geq \frac{\epsilon}{4(\gamma+\gamma^2M(\eta))}\\
        (\text{Lemma \ref{lem_spec_L_blk}})&\;\\
        \Ra &\;\|L-L_{\text{blk}}\|\geq \frac{\epsilon}{8\sqrt{k}(\gamma+\gamma^2M(\eta))}\Delta n_{\min}\\
        \Ra &\;\|L-L_{\text{blk}}\|\geq 8\sqrt{kn_{\max}\tilde{b}_{\max}\log (24n/\delta)}\,.
    \end{align*}
    That is, for a given $\delta\geq 4n^{-c}$, if \eqref{eq_pf_thm_err_cond1}\eqref{eq_pf_thm_err_cond3} hold, then
    \begin{align*}
        &\;\prob\lp \|V_k-\hat{V}_k\|\geq \frac{\epsilon}{4(\gamma+\gamma^2M(\eta))}, \text{``SC returns true $\{\mathcal{I}_i\}_{i=1}^k$"}\rp\\
        &\;\quad\quad\quad\quad\leq \prob\lp \|L-L_{\text{blk}}\|\geq 8\sqrt{kn_{\max}\tilde{b}_{\max}\log (24n/\delta)}\rp\leq \frac{\delta}{6}\,.
    \end{align*}
    Now by Lemma \ref{lem_err_prob_bd}, we have
    \begin{align*}
        &\;\prob\lp \sup_{s\in(-j\eta,+j\eta)}\|T_{yu}(s)-\hat{T}_k(s)\|\geq \epsilon\rp\\
        \leq &\;2\prob\lp \lambda_{k+1}(L)\leq \frac{1}{F_l(\eta)}\lp \frac{2}{\epsilon}(\gamma M(\eta) +1)^2+M(\eta)+\gamma M^2(\eta)\rp\rp\\
        &\; \qquad\qquad\qquad\qquad\qquad+\prob\lp \|V_k-\hat{V}_k\|\geq \frac{\epsilon}{4(\gamma+\gamma^2M(\eta))}\rp\\
        \leq &\;2\prob\lp \lambda_{k+1}(L)\leq \frac{1}{F_l(\eta)}\lp \frac{2}{\epsilon}(\gamma M(\eta) +1)^2+M(\eta)+\gamma M^2(\eta)\rp\rp\\
        &\; \qquad\qquad\qquad+\prob\lp \|V_k-\hat{V}_k\|\geq \frac{\epsilon}{4(\gamma+\gamma^2M(\eta))}, \text{``SC returns true $\{\mathcal{I}_i\}_{i=1}^k$"}\rp\\
        &\; \qquad\qquad\qquad+\prob\lp \|V_k-\hat{V}_k\|\geq \frac{\epsilon}{4(\gamma+\gamma^2M(\eta))}, \text{``SC does not return true $\{\mathcal{I}_i\}_{i=1}^k$"}\rp\\
        \leq &\;2\prob\lp \lambda_{k+1}(L)\leq \frac{1}{F_l(\eta)}\lp \frac{2}{\epsilon}(\gamma M(\eta) +1)^2+M(\eta)+\gamma M^2(\eta)\rp\rp\\
        &\; \qquad\qquad\qquad+\prob\lp \|V_k-\hat{V}_k\|\geq \frac{\epsilon}{4(\gamma+\gamma^2M(\eta))}, \text{``SC returns true $\{\mathcal{I}_i\}_{i=1}^k$"}\rp\\
        &\; \qquad\qquad\qquad+\prob\lp \text{``SC does not return true $\{\mathcal{I}_i\}_{i=1}^k$"}\rp\\
        \leq &\; 2\cdot \frac{\delta}{6}+\frac{\delta}{6}+\frac{\delta}{2}=\delta\,,
    \end{align*}
    In the last inequality, we upper bound the first and the second probability by picking a sufficiently large $n$ such that \eqref{eq_pf_thm_err_cond1}\eqref{eq_pf_thm_err_cond2}\eqref{eq_pf_thm_err_cond3} hold, and the last probability is upper bounded by $\frac{\delta}{2}$ if we pick $n\geq \tilde{N}\lp\frac{\delta}{2}\rp$ by our assumption \ref{assump_sc_algo}.
\end{proof}
\subsection{Proofs of Auxiliary Lemmas}\label{app_ssec_pf_ax_lem}

\begin{proof}[Proof of Lemma \ref{lem_Tk_hinf_bd}]
    For each $g_i(s), i=1,\cdots,n$, we have, by the OSP property,
    $$
        Re(g_i(s))\geq \frac{1}{\gamma} |g_i(s)|^2,\forall Re(s)>0\,.
    $$
    we have, for the diagonal transfer matrix $G(s)=\dg\{g_i(s)\}_{i=1}^n$:
    \be
        2Re(G(s))=G^*(s)+G(s)\succeq \frac{2}{\gamma}   G^*(s)G(s)\,, \forall Re(s)>0\,. \label{eq_lem_Tk_hinf_bd_eq1}
    \ee
    Since $g_i(s)$ are all OSP, then all $g_i(s)$ are positive real~\citep{khalil1996nonlinear}. A positive real function that is not a zero function has no zero nor pole on the left half plane. Therefore $g_i(s)$ are invertible for all $Re(s)>0$, which ensures that $G(s)$ is invertible for all $Re(s)>0$. Multiply $(G^{-1})^*$ on the left and $G^{-1}$ on the right of \eqref{eq_lem_Tk_hinf_bd_eq1}, we have
    \be
        G^{-1}(s)+(G^{-1}(s))^*\succeq \frac{2}{\gamma}I, \forall Re(s)>0\,. \label{eq_lem_Tk_hinf_bd_eq2}
    \ee
    Multiply $V^T_k$ on the left and $V_k$ on the right of \eqref{eq_lem_Tk_hinf_bd_eq2}, we have
    $$
        V^T_kG^{-1}(s)V_k+(V_k^TG^{-1}(s)V_k)^*\succeq \frac{2}{\gamma}I, \forall Re(s)>0\,,
    $$
    using the fact that $f(s)\Lambda_k$ is PR, we have 
    \be
        V^T_kG^{-1}(s)V_k+f(s)\Lambda_k+(V_k^TG^{-1}(s)V_k+f(s)\Lambda_k)^*\succeq \frac{2}{\gamma} I, \forall Re(s)>0\,,
    \ee
    Notice that we have defined $H_k(s)=V^T_kG^{-1}(s)V_k+f(s)\Lambda_k$, then we conclude that
    $
        H_k(s)+(H_k(s))^*\succeq\frac{2}{\gamma} I\,,
    $
    or equivalently,
    \be
        \bmt I \\ H_k(s)\emt^* \bmt -\frac{2}{\gamma} I & I\\ I & 0\emt\bmt I \\ H_k(s)\emt\succeq 0, \forall Re(s)>0
    \ee
    Moreover, we have
    $$
        \bmt -\frac{2}{\gamma} I & I\\ I & 0\emt+\bmt \frac{1}{\gamma} I & 0\\ 0 & -\gamma^2 \frac{\epsilon}{2} I\emt=\bmt -\frac{1}{\gamma} I & I\\ I & -\gamma I\emt\preceq 0\,,
    $$
    since its Schur complement is a zero matrix.
    
    Therefore,
    $$
        \bmt I \\ H_k(s)\emt^* \bmt -\frac{1}{\gamma} I & I\\ I & -\gamma I\emt\bmt I \\ H_k(s)\emt
        \succeq\bmt I \\ H_k(s)\emt^* \bmt -\frac{2}{\gamma} I & I\\ I & 0\emt\bmt I \\ H_k(s)\emt\succeq 0\,, \forall Re(s)>0\,,
    $$
    which is exactly,
    $$\gamma(H_k(s))^*(H_k(s))\succeq \frac{1}{\gamma}I\,, \forall Re(s)>0\,.$$
    This shows that
    $$
        \sigma_{\min}^2(H_k(s))\geq \frac{1}{\gamma^2},\forall Re(s)>0\,,
    $$
    which leads to
    $$
        \|T_k(s)\|=\|V_kH_k^{-1}(s)V_k^T\|=\|H_k^{-1}(s)\|\leq \gamma\,, \forall Re(s)>0\,.
    $$
    This is exactly $\|T_k(s)\|_{\mathcal{H}_\infty}\leq \gamma$.
\end{proof}

\begin{proof}[Proof of Lemma \ref{lem_T_k_approx_err}]
    Denote $H_k(s):=V_k^T\dg\{g_i^{-1}(s)\}V_k+f(s)\Lambda_k$ and $\hat{H}_k(s):=\hat{V}_k^T\dg\{g_i^{-1}(s)\}\hat{V}_k+f(s)\Lambda_k$, then
    \begin{align*}
        &\;\|T_k(s)-\hat{T}_k(s)\|\\
        =&\;\|V_kH_k(s)V_k^T-\hat{V}_k\hat{H}_k(s)\hat{V}_k^T\|\\
        =&\;\|\hat{V}_kH_k^{-1}(s)(V_k^T-\hat{V}_k^T)+(V_k-\hat{V}_k)H_k^{-1}(s)V_k^T+\hat{V}_k(H_k^{-1}(s)-\hat{H}_k^{-1}(s))\hat{V}_k^T\|\\
        \leq &\; \|\hat{V}_kH_k^{-1}(s)(V_k^T-\hat{V}_k^T)\|+\|(V_k-\hat{V}_k)H_k^{-1}(s)V_k^T\|+\|\hat{V}_k(H_k^{-1}(s)-\hat{H}_k^{-1}(s))\hat{V}_k^T\|
    \end{align*}
    Notice that for any $s\in (-j\eta,+j\eta)$,
    \be\label{eq_lem_T_k_approx_err_eq1}
        \|\hat{V}_kH_k^{-1}(s)(V_k^T-\hat{V}_k^T)\|\leq \|H_k^{-1}(s)\|\|V_k-\hat{V}_k\|\leq \|H_k^{-1}(s)\|_{\mathcal{H}_\infty}\|V_k-\hat{V}_k\|\leq \gamma\|V_k-\hat{V}_k\|\,,
    \ee
    where the last inequality uses the intermediate result $\|H_k^{-1}(s)\|$ in the proof for Lemma \ref{lem_Tk_hinf_bd}. Similarly,
    \be\label{eq_lem_T_k_approx_err_eq2}
        \|(V_k-\hat{V}_k)H_k^{-1}(s)V_k^T\|\leq \gamma \|V_k-\hat{V}_k\|\,.
    \ee
    For the last term, we have for any $s\in (-j\eta,+j\eta)$,
    \begin{align}
        &\;\|\hat{V}_k(H_k^{-1}(s)-\hat{H}_k^{-1}(s))\hat{V}_k^T\|\nonumber\\
        \leq &\; \|H_k^{-1}(s)-\hat{H}_k^{-1}(s)\|\nonumber\\
        =&\; \|\hat{H}_k^{-1}(s)(\hat{H}_k(s)-H_k(s))H_k^{-1}(s)\|\nonumber\\
        \leq &\; \|\hat{H}_k^{-1}(s)\|\|H_k^{-1}(s)\|\|\hat{H}_k(s)-H_k(s)\|\nonumber\\
        \leq &\; \gamma^2 \|V_k^T\dg\{g_i^{-1}(s)\}V_k-\hat{V}_k^T\dg\{g_i^{-1}(s)\}\hat{V}_k\|\nonumber\\
        \leq &\; \gamma^2\|(V_k^T-\hat{V}_k^T)\dg\{g_i^{-1}(s)\}V_k+\hat{V}_k^T\dg\{g_i^{-1}(s)\}(V_k-\hat{V}_k)\|\nonumber\\
        \leq &\; 2\gamma^2\|\dg\{g_i^{-1}(s)\}\|\|V_k-\hat{V}_k\|\leq 2\gamma^2M(\eta)\|V_k-\hat{V}_k\|\,.\label{eq_lem_T_k_approx_err_eq3} 
    \end{align}
    Using the bounds in \eqref{eq_lem_T_k_approx_err_eq1}\eqref{eq_lem_T_k_approx_err_eq2}\eqref{eq_lem_T_k_approx_err_eq3}, we finally have
    \begin{align*}
        &\;\sup_{s\in(-j\eta,+j\eta)}\|T_k(s)-\hat{T}_k(s)\|\\
        \leq &\; \sup_{s\in(-j\eta,+j\eta)}\lp \|\hat{V}_kH_k^{-1}(s)(V_k^T-\hat{V}_k^T)\|+\|(V_k-\hat{V}_k)H_k^{-1}(s)V_k^T\|+\|\hat{V}_k(H_k^{-1}(s)-\hat{H}_k^{-1}(s))\hat{V}_k^T\|\rp\\
        \leq & \; 2(\gamma+\gamma^2M(\eta))\|V_k-\hat{V}_k\|\leq 2(\gamma+\gamma^2M(\eta))\|V_k-\hat{V}_k\|_F\,.
    \end{align*}
\end{proof}
\begin{proof}[Proof of Lemma \ref{lem_err_prob_bd}]
    We have defined $T_k(s)=V_k(V_k^T\dg\{g_i^{-1}(s)\}V_k+f(s)\Lambda_k)^{-1}V_k^T$, then
    \begin{align}
        &\;\prob\lp \sup_{s\in(-j\eta,+j\eta)}\|T_{yu}(s)-\hat{T}_k(s)\|\geq \epsilon\rp\nonumber\\
        \leq &\;\prob\lp \sup_{s\in(-j\eta,+j\eta)}\|T_{yu}(s)-T_k(s)\|+ \sup_{s\in(-j\eta,+j\eta)}\|T_k(s)-\hat{T}_k(s)\|\geq \epsilon\rp\nonumber\\
        \leq &\; \prob\lp \sup_{s\in(-j\eta,+j\eta)}\|T_{yu}(s)-T_k(s)\|\geq \frac{\epsilon}{2}\rp+\prob\lp \sup_{s\in(-j\eta,+j\eta)}\|T_k(s)-\hat{T}_k(s)\|\geq \frac{\epsilon}{2}\rp\,.\label{eq_lem_err_prob_bd_eq1}
    \end{align}
    For the first term, we have
    \begin{align}
        &\;\prob\lp \sup_{s\in(-j\eta,+j\eta)}\|T_{yu}(s)-T_k(s)\|\geq \frac{\epsilon}{2}\rp\nonumber\\
        =&\;\prob\lp \sup_{s\in(-j\eta,+j\eta)}\|T_{yu}(s)-T_k(s)\|\geq \frac{\epsilon}{2},\ \lambda_{k+1}(L)\leq \frac{M(\eta)+\gamma M^2(\eta)}{F_l(\eta)}\rp\nonumber\\
        &\; \qquad + \prob\lp \sup_{s\in(-j\eta,+j\eta)}\|T_{yu}(s)-T_k(s)\|\geq \frac{\epsilon}{2},\ \lambda_{k+1}(L)> \frac{M(\eta)+\gamma M^2(\eta)}{F_l(\eta)}\rp\nonumber\\
        \leq &\; \prob\lp  \lambda_{k+1}(L)\leq \frac{M(\eta)+\gamma M^2(\eta)}{F_l(\eta)}\rp\nonumber\\
        &\;\qquad+\prob\lp \sup_{s\in(-j\eta,+j\eta)}\|T_{yu}(s)-T_k(s)\|\geq \frac{\epsilon}{2},\ \lambda_{k+1}(L)> \frac{M(\eta)+\gamma M^2(\eta)}{F_l(\eta)}\rp\nonumber\\
        \overset{(a)}{\leq} &\; \prob\lp  \lambda_{k+1}(L)\leq \frac{M(\eta)+\gamma M^2(\eta)}{F_l(\eta)}\rp+\prob\lp \frac{(\gamma M(\eta)+1)^2}{F_l(\eta) \lambda_{k+1}(L)-M(\eta)-\gamma M(\eta)}\geq \frac{\epsilon}{2}\rp\nonumber\\
        \overset{(b)}{\leq} &\;2\prob\lp \lambda_{k+1}(L)\leq \frac{1}{F_l(\eta)}\lp \frac{2}{\epsilon}(\gamma M(\eta) +1)^2+M(\eta)+\gamma M^2(\eta)\rp\rp\,,\label{eq_lem_err_prob_bd_eq2}
    \end{align}
    where (a) is from the fact that when $\lambda_{k+1}(L)> \frac{M(\eta)+\gamma M^2(\eta)}{F_l(\eta)}$, we can apply Theorem \ref{thm_Tk} for any $s_0\in (-j\eta,+j\eta)$, with a uniform bound
    $$
        \|T_{yu}(s_0)-T_k(s_0)\|\leq \frac{(\gamma M(\eta)+1)^2}{F_l(\eta) \lambda_{k+1}(L)-M(\eta)-\gamma M(\eta)}\,,
    $$
    then applying supremum gives us $
        \sup_{s\in(-j\eta,+j\eta)}\|T_{yu}(s)-T_k(s)\|\leq \frac{(\gamma M(\eta)+1)^2}{F_l(\eta) \lambda_{k+1}(L)-M(\eta)-\gamma M(\eta)}\,,
    $ which implies the event $\frac{(\gamma M(\eta)+1)^2}{F_l(\eta) \lambda_{k+1}(L)-M(\eta)-\gamma M(\eta)}\geq \frac{\epsilon}{2}$. And the (b) is due to the fact that the second probability is always larger than the first one.
    
    For the second term, we have, by Lemma \ref{lem_T_k_approx_err}
    \begin{align}
        \prob\lp \sup_{s\in(-j\eta,+j\eta)}\|T_k(s)-\hat{T}_k(s)\|\geq \frac{\epsilon}{2}\rp &\leq\; \prob \lp 2(\gamma+\gamma^2M(\eta))\|V_k-\hat{V}_k\|\geq \frac{\epsilon}{2}\rp\nonumber\\
        &=\; \prob\lp \|V_k-\hat{V}_k\|\geq \frac{\epsilon}{4(\gamma+\gamma^2M(\eta))}\rp\,.\label{eq_lem_err_prob_bd_eq3}
    \end{align}
    Apply \eqref{eq_lem_err_prob_bd_eq2}\eqref{eq_lem_err_prob_bd_eq3} to \eqref{eq_lem_err_prob_bd_eq1} gives the desired bound.
\end{proof}
\begin{proof}[Proof of Lemma \ref{lem_fro_bd_sin}]
    Consider $V_k=\bmt \frac{\one}{\sqrt{n}}& v_2(L) &\cdots & v_k(L)\emt$ from the random Laplacian matrix, and $V_k^\text{blk}=\bmt \frac{\one}{\sqrt{n}}& v_2(L_\text{blk}) &\cdots & v_k(L_\text{blk})\emt$ from $L_\text{blk}$. The singular values of $V_k^TV_k^{\text{blk}}$ are exactly the cosine of the principal angles between the two subspace spanned by $V_k$ and $V_k^{\text{blk}}$. Notice that 
    \be
        V_k^TV_k^{\text{blk}}=\bmt 1 & 0\\ 0 & \tilde{V}_k^T\tilde{V}_k^{\text{blk}}\emt,\text{where } \tilde{V}_k\!=\!\bmt v_2(L) &\cdots & v_k(L)\emt, \tilde{V}_k^{\text{blk}}\!=\!\bmt  v_2(L_\text{blk}) &\cdots & v_k(L_\text{blk})\emt\,,
    \ee
    then there exists $O_1,O_2\in\mathcal{O}^{(k-1)\times(k-1)}$ such that
    \be
        \bmt 1&0\\ 0& O_1^T\emt\bmt 1 & 0\\ 0 & \tilde{V}_k^T\tilde{V}_k^{\text{blk}}\emt\bmt 1&0\\ 0& O_2\emt=\bmt 1 & & & \\ & \cos\theta_2 & & \\ & & \ddots & \\ & & & \cos\theta_k \emt\,,
    \ee
    where the first diagonal term corresponds to the cosine of the first principal angle $\theta_1=0$. Now consider the orthogonal matrix $O=\bmt 1&0\\ 0& O_2\emt\bmt 1&0\\ 0& O_1^T\emt\in\mathbb{R}^{k\times k}$, we have
    \begin{align*}
        \|V_k^{\text{blk}}O-V_k\|_F^2&=\;\tr\lp (V_k^{\text{blk}}O-V_k)^T(V_k^{\text{blk}}O-V_k)\rp\\
        &=\;2d-2\tr(V_k^TV_k^{\text{blk}}O)\\
        &=\;2d-2\tr\lp \bmt 1&0\\ 0& O_1^T\emt V_k^TV_k^{\text{blk}}\bmt 1&0\\ 0& O_2\emt\rp\\
        &=\; 2d-\sum_{i=1}^k\cos\theta_i\\
        &\leq\; 2d-\sum_{i=1}^k\cos^2\theta_i =\|\sin\Theta(V_k,V_k^{\text{blk}})\|^2_F\,.
    \end{align*}
    Now consider the $\hat{V}_k=P_{\{\mathcal{I}_i\}_{i=1}^k}S^*$, where $\{\mathcal{I}_i\}_{i=1}^k$ is the clustering result from the spectral clustering algorithm, and $S^*$ from solving the optimization problem in \eqref{eq_V_fit_opt_fix_part}. Notice that if the clustering result $\{\mathcal{I}_i\}_{i=1}^k$ is correct, i.e., $L_\text{blk}$ is indeed $k$-block-ideal w.r.t. $P_{\{\mathcal{I}_i\}_{i=1}^k}$, then there exists some invertible matrix $\tilde{S}\in\mathbb{R}^{k\times k}$ such that \be
        V_k^{\text{blk}}=P_{\{\mathcal{I}_i\}_{i=1}^k}\tilde{S}\,,
    \ee
    which implicitly requires $\tilde{S}e_1=\one_k$ and $\tilde{S}^T\dg\{n_i\}_{i=1}^k\tilde{S}=I_k$. It is easy to show that
    $S=\tilde{S}O$ is a feasible solution to \eqref{eq_V_fit_opt_fix_part}:
    \be
        Se_1=\tilde{S}Oe_1=\tilde{S}e_1=\one_k,\ S^T\dg\{n_i\}_{i=1}^kS=O^T\tilde{S}^T\dg\{n_i\}_{i=1}^k\tilde{S}O=I_k\,.
    \ee
    Therefore
    \be
        \|\hat{V}_k-V_k\|_F\leq \|P_{\{\mathcal{I}_i\}_{i=1}^k}S-V_k\|_F=\|V_k^{\text{blk}}O-V_k\|_F\leq \|\sin \Theta(V_k,V_k^{\text{blk}})\|_F\,.
    \ee
\end{proof}
\newpage
\section{Eigenvalues and eigenvectors of $L_\text{blk}$}\label{app_spec_L_blk}
Assume the network has the following adjacency matrix:
    \be
        A_\text{blk}=\underbrace{\bmt \one_{n_1} & 0 & \cdots & 0\\
        0 & \one_{n_2} & \cdots & 0\\
        \vdots & & \ddots & \vdots \\
        0 & 0 & \cdots & \one_{n_k}\emt}_{:=P} B_k P^T\,,
    \ee
    where $A_k\in\mathbb{R}^{k\times k}$ and
    \be
        [B_k]_{ij}=\begin{cases}
            \alpha_{ii}, & i=j\\
            \beta_{ij}, & i\leq j\\
            \beta_{ij}, & i>j
        \end{cases}\,.
    \ee
    The Laplacian
    \be
        L_\text{blk}=\dg\{A_\text{blk}\one_n\}-A_\text{blk}\,.
    \ee
    
\begin{proposition}\label{prop_eigenspace_L_blk}
        Let $n_{\min}=\min\{n_i:\ i=1,\cdots,k\}$, and $n_{\max}=\max\{n_i:\ i=1,\cdots,k\}$ Suppose
        \be
            \min_i\{\alpha_{ii}\}-\frac{2n_{\max}}{n_{\min}}\max_ i\sum_{j\neq i}\beta_{ij}:=\Delta >0\,.
        \ee
        Define
        \be
            \tilde{L}_k=\dg\{\Tilde{B}_{k}\one_k\}-\Tilde{A}_{k},\quad\Tilde{B}_{k}=B_k\cdot\dg\{n_i\}_{i=1}^k\,,
        \ee
        and let $v_i(\tilde{L}_k)$ be the right eigenvector of $\tilde{L}_k$ associated with $\lambda_i(\tilde{L}_k)$. Then for $i=1,\cdots,k$, we have
        \begin{enumerate}
            \item ($L_\text{blk}$ is k-block-ideal)
            \be
                \lambda_i(L_{\text{blk}})=\lambda_i(\tilde{L}_k), v_i(L_{\text{blk}})= \bmt \one_{n_1} & 0 & \cdots & 0\\
                0 & \one_{n_2} & \cdots & 0\\
                \vdots & & \ddots & \vdots \\
                0 & 0 & \cdots & \one_{n_k}\emt v_i(\tilde{L}_k)\,.
            \ee
        \end{enumerate}
        Moreover, let  $b_{\min}:=\min\{[B_k\one_k]_i:\ i=1,\cdots,k\}$ be the minimum row sum of $B_k$, then we have
        \begin{enumerate}[resume]
            \item ($\lambda_{k+1}(L_{\text{blk}})$ is large)
            \be
            \lambda_{k+1}(L_{\text{blk}})\geq b_{\min}n_{\min}\,,
        \ee
        \item (there is a sufficient spectral gap $\lambda_{k+1}(L_{\text{blk}})-\lambda_{k}(L_{\text{blk}})$)
        \be
            \lambda_{k+1}(L_{\text{blk}})-\lambda_{k}(L_{\text{blk}})\geq \Delta n_{\min}\,.
        \ee
        \end{enumerate}
        
    \end{proposition}    
    \begin{proof}
        The proof takes few steps: First we show that $\{(\lambda_i(\tilde{L}_k),v_i(\tilde{L}_k))\}_{i=1}^k$ are eigenpairs of $L_\text{blk}$, then we show that $\{\lambda_i(\tilde{L}_k)\}_{i=1}^k$ are indeed the first $k$ smallest eigenvalues of $L_\text{blk}$. Lastly we provide the lower bound on both  $\lambda_{k+1}(L_{\text{blk}})$ and $\lambda_{k+1}(L_{\text{blk}})-\lambda_{k}(L_{\text{blk}})$\,.
        \paragraph{Show eigenpairs$\{(\lambda_i(\tilde{L}_k),v_i(\tilde{L}_k))\}_{i=1}^k$}: Notice that
        \begin{align}
            L_{\text{blk}}P&=\;(\dg\{A_\text{blk}\one_n\}-A_\text{blk})P\nonumber\\
            &=\; (\dg\{PB_kP^T\one_n\}-PB_kP^T)P\nonumber\\
            &=\; \dg\{PB_k\dg\{n_i\}_{i=1}^k\one_k\}-PB_k\dg\{n_i\}_{i=1}^k\nonumber\\
            &=\;P\dg\{B_k\dg\{n_i\}_{i=1}^k\one_k\}-PB_k\dg\{n_i\}_{i=1}^k=P\tilde{L}_\text{blk}\,,\label{eq_prop_eigenspace_L_blk_eq1}
        \end{align}
        where we used an equality $P\dg\{x\}=\dg\{Px\}$ for any $x\in\mathbb{R}^{k}$ due to the special structure of $P$. We can obtain $k$ eigenpairs through \eqref{eq_prop_eigenspace_L_blk_eq1}: Given any eigenpair $(\lambda_i(\tilde{L}_\text{blk}),v_i(\tilde{L}_k))$, we have
        \be
            L_\text{blk}Pv_i(\tilde{L}_k)=P\tilde{L}_kv_i(\tilde{L}_k)=\lambda_i(\tilde{L}_\text{blk})Pv_i(\tilde{L}_k)\,,
        \ee
        which suggests $(\lambda_i(\tilde{L}_k),Pv_i(\tilde{L}_k))$ is an eigenpair of $L_\text{blk}$. This holds for every $i=1,\cdots,k$.
        \paragraph{Show that  $\{\lambda_i(\tilde{L}_k)\}_{i=1}^k$ are the first $k$ smallest eigenvalues}: The remaining eigenvalues of $L_\text{blk}$ are easy to find:
        \begin{itemize}[leftmargin=4.5mm]
            \item \!\!\! $\lp \lambda=n_1\alpha_{11}+\sum_{j\neq 1}\beta_{1j}n_j, v=\bmt v_1 \\ 0 \\
            \vdots \\0 \emt\rp$ is an eigenpair for any $v_1\in\mathbb{S}^{n_1-1}$ such that $\one_{n_1}^Tv_1\!=\!0$.
            \item \!\!\! $\lp \lambda=n_2\alpha_{22}+\sum_{j\neq 2}\beta_{2j}n_j, v=\bmt 0 \\ v_2 \\
            \vdots \\0 \emt\rp$ is an eigenpair for any $v_2\in\mathbb{S}^{n_2-1}$ such that $\one_{n_2}^Tv_2\!=\!0$.
            \item $\cdots$
            \item \!\!\!\! $\lp \lambda=n_k\alpha_{kk}+\sum_{j\neq k}\beta_{kj}n_j, v=\bmt 0 \\ 0 \\
            \vdots \\ v_k \emt\rp$ is an eigenpair for any $v_k\in\!\mathbb{S}^{n_k-1}$ such that $\one_{n_k}^Tv_k\!=\!0$.
        \end{itemize}
        Any choice of such $(\lambda,v)$ is an eigenpair because, for example, the eigenpair associated with some $v_1$ satisfies
        \begin{align*}
            L_\text{blk}\bmt v_1 \\ 0 \\
            \vdots \\0 \emt\! \!&=\!\!\lp\bmt (n_1\alpha_{11}+\sum_{j\neq 1}\beta_{1j}n_j)I_{n_1} & & \\
            &\ddots & \\
            & &  (n_k\alpha_{kk}+\sum_{j\neq k}\beta_{kj}n_j)I_{n_k}\emt-PB_kP^T\rp\!\bmt v_1 \\ 0 \\
            \vdots \\0 \emt\\
            &= \bmt (n_1\alpha_{11}+\sum_{j\neq 1}\beta_{1j}n_j)I_{n_1} & & \\
            &\ddots & \\
            & &  (n_k\alpha_{kk}+\sum_{j\neq k}\beta_{kj}n_j)I_{n_k}\emt\bmt v_1 \\ 0 \\
            \vdots \\0 \emt\\
            &= \lp n_1\alpha_{11}+\sum_{j\neq 1}\beta_{1j}n_j\rp v_1\,.
        \end{align*}
        Similar argument can be made for other pairs. This gives us all the rest of the eigenvalues: each eigenvalue $n_i\alpha_{ii}+\sum_{j\neq i}\beta_{ij}n_j$ has multiplicity $n_i-1$. Together with the $k$ eigenvalues $\{\lambda_i(\tilde{L}_k)\}_{i=1}^k$ we have already found in previous derivation, we have all the eigenvalues of $L_\text{blk}$.
        
        The claim that $\{\lambda_i(\tilde{L}_k)\}_{i=1}^k$ are the first $k$ smallest eigenvalues is shown by our assumption:
        \begin{align}
            \min_{i}\lp n_i\alpha_{ii}+\sum_{j\neq i}\beta_{ij}n_j\rp&\geq\; \min_i n_i\alpha_{ii}\nonumber\\
            &\geq\; n_{\min} \min_i\alpha_{ii}\nonumber\\ &\geq\;2n_{\max}\max_ i\sum_{j\neq i}\beta_{ij}+n_{\min}\Delta\nonumber\\
            &\geq \; \max_i (n_i+n_{\max})\sum_{j\neq i}\beta_{ij}+n_{\min}\Delta\nonumber\\
            &\geq \;\max_i \lp n_i\sum_{j\neq i}\beta_{ij} + \sum_{j\neq i}\beta_{ij}n_j\rp+n_{\min}\Delta\nonumber\\
            &\geq\; \max_i\lambda_i(\tilde{L}_k)+n_{\min}\Delta\,,\label{eq_prop_eigenspace_L_blk_eq2}
        \end{align}
        where the last inequality is from the Gershgorin disk theorem~\citep{Horn:2012:MA:2422911} by noticing that for $i$-th column of $\tilde{L}_k$, the diagonal term is $\sum_{j\neq i}\beta_{ij}n_j$ and the sum of the absolute value of the off-diagonal terms is $n_i\sum_{j\neq i}\beta_{ij}$. \eqref{eq_prop_eigenspace_L_blk_eq2} is more than enough to show $\{\lambda_i(\tilde{L}_k)\}_{i=1}^k$ are the first $k$ smallest eigenvalues of $L_{\text{blk}}$.
        
        \paragraph{Bound on the eigenvalue and spectral gap}
        
        Knowing $\{\lambda_i(\tilde{L}_k)\}_{i=1}^k$ are the first $k$ smallest eigenvalues, we have
        \begin{align*}
            \lambda_{k+1}(L_\text{blk})=\min_{i}\lp n_i\alpha_{ii}+\sum_{j\neq i}\beta_{ij}n_j\rp\geq \min_i\lp \alpha_{ii}+\sum_{j=i}\beta_{ij}\rp n_{\min} = b_{\min}n_{\min}\,,
        \end{align*}
        and \eqref{eq_prop_eigenspace_L_blk_eq2} already shows
        \begin{align*}
            \lambda_{k+1}(L_\text{blk})=\min_{i}\lp n_i\alpha_{ii}+\sum_{j\neq i}\beta_{ij}n_j\rp\geq \max_i\lambda_i(\tilde{L}_k)+n_{\min}\Delta=\lambda_{k}(L_{\text{blk}})+n_{\min}\Delta\,.
        \end{align*}
    \end{proof}
    